\documentclass[english]{paper}

\usepackage[T1]{fontenc}
\usepackage[latin9]{inputenc}
\usepackage{geometry}
\usepackage{xcolor}
\usepackage{pdfcolmk}
\usepackage{babel}
\usepackage{units}
\usepackage{bm}
\usepackage{amsthm}
\usepackage{amsmath}
\usepackage{amssymb}
\usepackage{graphicx}
\usepackage{setspace}
\usepackage{esint}
\usepackage[authoryear]{natbib}
\PassOptionsToPackage{normalem}{ulem}
\usepackage{ulem}
\usepackage{rotating}
\usepackage[unicode=true,pdfusetitle,
 bookmarks=true,bookmarksnumbered=false,bookmarksopen=false,
 breaklinks=false,pdfborder={0 0 1},backref=false,colorlinks=false]
 {hyperref}

\usepackage{wrapfig}

\usepackage{enumerate}

\makeatletter

\theoremstyle{plain}
\newtheorem{thm}{\protect\theoremname}
  \theoremstyle{plain}
  \newtheorem{assumption}{\protect\assumptionname}
  \theoremstyle{definition}
  \newtheorem{defn}{\protect\definitionname}
  \theoremstyle{plain}
  \newtheorem{cor}{\protect\corollaryname}
  \theoremstyle{plain}
  \newtheorem{lyxalgorithm}{\protect\algorithmname}
  \theoremstyle{plain}
  \newtheorem{prop}{\protect\propositionname}
  \theoremstyle{plain}
  \newtheorem{lem}{\protect\lemmaname}

\usepackage{babel}
\bibpunct{(}{)}{,}{a}{,}{,}

\usepackage{url}
\usepackage{bm}
\usepackage[shortlabels]{enumitem}
\setlist[enumerate,1]{label=(\alph*)}
\setlist[enumerate,2]{label=(\roman*)}

  \providecommand{\algorithmname}{Algorithm}
  \providecommand{\assumptionname}{Assumption}
  \providecommand{\corollaryname}{Corollary}
  \providecommand{\definitionname}{Definition}
  \providecommand{\propositionname}{Proposition}
\providecommand{\theoremname}{Theorem}

\@ifundefined{showcaptionsetup}{}{%
 \PassOptionsToPackage{caption=false}{subfig}}
\usepackage{subfig}
\makeatother

  \addto\captionsenglish{\renewcommand{\algorithmname}{Algorithm}}
  \addto\captionsenglish{\renewcommand{\assumptionname}{Assumption}}
  \addto\captionsenglish{\renewcommand{\corollaryname}{Corollary}}
  \addto\captionsenglish{\renewcommand{\definitionname}{Definition}}
  \addto\captionsenglish{\renewcommand{\lemmaname}{Lemma}}
  \addto\captionsenglish{\renewcommand{\propositionname}{Proposition}}
  \addto\captionsenglish{\renewcommand{\theoremname}{Theorem}}
  \providecommand{\algorithmname}{Algorithm}
  \providecommand{\assumptionname}{Assumption}
  \providecommand{\corollaryname}{Corollary}
  \providecommand{\definitionname}{Definition}
  \providecommand{\lemmaname}{Lemma}
  \providecommand{\propositionname}{Proposition}
\providecommand{\theoremname}{Theorem}

\begin{document}
\global\long\def\CC{\mathbb{C}}
 \global\long\def\SS{S^{1}}
 \global\long\def\RR{\mathbb{R}}
 \global\long\def\actson{\curvearrowright}
 \global\long\def\ra{\rightarrow}
 \global\long\def\z{\mathbf{z}}
 \global\long\def\ZZ{\mathbb{Z}}
 \global\long\def\NN{\mathbb{N}}
 \global\long\def\sgn{\mathrm{sgn}\:}
 \global\long\def\RRpos{\RR_{>0}}
 \global\long\def\var{\mathrm{var}}
 \global\long\def\circint{\int_{-\pi}^{\pi}}
 \global\long\def\F{\mathcal{F}}
 \global\long\def\pb#1{\langle#1\rangle}
 \global\long\def\op{\mathrm{op}}
 \global\long\def\Op{\mathrm{op}}
 \global\long\def\supp{\mathrm{supp}}
 \global\long\def\ceil#1{\lceil#1\rceil}
 \global\long\def\TV{\mathrm{TV}}
 \global\long\def\floor#1{\lfloor#1\rfloor}
 \global\long\def\vt{\vartheta}
 \global\long\def\vp{\varphi}
 \global\long\def\class#1{[#1]}
 \global\long\def\of{(\cdot)}
 \global\long\def\one{1}
 \global\long\def\cov{\mathrm{cov}}
 \global\long\def\CC{\mathbb{C}}
 \global\long\def\SS{S^{1}}
 \global\long\def\RR{\mathbb{R}}
 \global\long\def\actson{\curvearrowright}
 \global\long\def\ra{\rightarrow}
 \global\long\def\z{\mathbf{z}}
 \global\long\def\ZZ{\mathbb{Z}}
 \global\long\def\h{\mu}
 \global\long\def\convr{*_{\RR}}
 \global\long\def\x{\mathbf{x}}
 \global\long\def\ve{\epsilon}
 \global\long\def\cv{\mathfrak{c}}
 \global\long\def\wh#1{\hat{#1}}
 \global\long\def\norm#1{\left|\left|#1\right|\right|}
 \global\long\def\Tmean{T}
 \global\long\def\Tslope{m}
 \global\long\def\degC#1{#1^{\circ}\mathrm{C}}
\global\long\def\X{\mathbf{X}}
\global\long\def\bbeta{\boldsymbol{\beta}}
\global\long\def\b{\mathbf{b}}
\global\long\def\Y{\mathbf{Y}}
\global\long\def\e{\bm{\epsilon}}
\global\long\def\s{\mathbf{s}}
\global\long\def\t{\mathbf{t}}
\global\long\def\R{\mathbf{R}}
\global\long\def\dAc{\partial A_c}
\global\long\def\cl{\mathrm{cl}}

\title{Confidence regions for excursion sets in asymptotically Gaussian
random fields, with an application to climate }

\author{Max Sommerfeld$^{1}$, Stephen Sain$^{2}$, Armin Schwartzman$^{3}$ }

\institution{$^{1}$FBMS, Universität Göttingen\\
	$^{2}$The Climate Corporation\\
 $^{3}$Department of Statistics, North Carolina State University}
\maketitle
\begin{abstract}
The goal of this paper is to give confidence regions for the excursion
set of a spatial function above a given threshold from repeated noisy
observations on a fine grid of fixed locations. Given an asymptotically
Gaussian estimator of the target function, a pair of data-dependent
nested excursion sets are constructed that are sub- and super-sets
of the true excursion set, respectively, with a desired confidence.
 Asymptotic coverage probabilities are determined via a multiplier bootstrap method,
not requiring Gaussianity of the original data nor stationarity or
smoothness of the limiting Gaussian field. The method is used to determine
regions in North America where the mean summer  and winter  temperatures are expected
to increase by mid 21st century by more than 2 degrees Celsius. \end{abstract}

\section{Introduction}

\label{sec:Introduction}Our motivation comes from the following problem.
Faced with a global change in temperature over the globe within the
next century, it is important to assess which geographical regions
are particularly at risk of extreme temperature change. The data used
here, obtained from the North American Regional Climate Change Assessment
Program (NARCCAP) project \citep{Mearns2009,Mearns2012,Mearns2013},
consists of two sets of 29 spatially registered arrays of mean seasonal temperatures for summer (June-August) and winter (December-February) evaluated at a fine grid of fixed
locations 0.5 degrees in geographic longitude and latitude apart over North America over two time periods:
late 20th century (1971-1999) and mid 21st century (2041-2069). Specifically, the data was produced by the WRFG climate model \citep{Michalakes2004} using boundary conditions from the CGCM3 global model \citep{Flato2005}. 
We would like to determine the regions whose difference in mean summer
or winter temperature between the two periods is greater than the
$2^{\circ}C$ benchmark \citep{Rogelj2009,Anderson2011}.
However, the observed differences may be confounded by the natural
year-to-year temperature variability. Can we set confidence bounds
on such regions that reflect the year-to-year variability in the data?

Unlike the usual data setup of spatial statistics, the above data
setup is more similar to that of population studies in brain imaging,
where a difference map between two conditions is estimated from repeated
co-located
image observations at a fine spatial grid under those conditions (see e.g.
\citet{Worsley1996,Genovese2002,Taylor2007,Schwartzman2010}). The methods in this paper are inspired
by that kind of analysis.

In general, suppose that we observe $n$ random fields $Y_{i}(\s),i=1,\ldots,n$,
over a spatial domain $S$, modeled as realizations of a general linear
model indexed by $\s\in S$. The target function $\mu:S\ra\RR$ could
be one of the parameters in the model indexed by $\s$, in our
case the mean difference temperature field. With a proper design, fitting the linear model at each location $\s$ will produce a consistent
and asymptotically Gaussian estimator $\hat{\mu}_n:S\ra\RR$ as $n$ increases.
Asymptotically Gaussian estimators indexed by $\s$ also appear in nonparametric
density estimation and regression. In those settings
$n$ would be the number of sample points. 

Let $A_{c}$ be the excursion
set of $\mu$ above a fixed threshold $c$, defined as $A_{c}:=A_{c}(\mu):=\left\{ \s\in S:\mu(\s)\geq c\right\} $,
and denote the analog for $\wh{\mu}_{n}$ by $\wh A_{c}:=A_{c}(\wh{\mu}_{n})$.
We wish to obtain confidence regions $\wh A_{c}^{\pm}$ that are nested
in the sense that $\wh A_{c}^{+}\subset\wh A_{c}\subset\wh A_{c}^{-}$
and for which the probability that 
\begin{equation}
\wh A_{c}^{+}\subset A_{c}\subset\wh A_{c}^{-}\label{eq:inclusion}
\end{equation}
holds is asymptotically above a desired level, say $90\%$. The sets
$\wh A_{c}^{\pm}$ here are obtained as excursion sets of the standardized
observed field $\wh{\mu}$ and we  call them Coverage Probability
Excursion (CoPE) sets. Assuming that the estimated field $\wh{\mu}$
satisfies a central limit theorem (CLT), we show that the probability that \eqref{eq:inclusion} holds is given asymptotically by the
distribution of the supremum of the limiting Gaussian random field
on the boundary $\partial A_{c}$ of the true excursion set. Using a plug-in estimate for the unknown boundary, we propose
a simple and efficient multiplier bootstrap procedure \citep{Wu1986,Hardle1993,Mammen1992,Mammen1993}, that does not require estimating the unknown (not necessarily stationary) correlation function
of the limiting field. The validity of this procedure for very high-dimensional
data has recently been shown by \citet{Chernozhukov2013}. 

For illustration, Figure \ref{data_res} shows CoPE sets for the temperature
data. The regions within the red boundary ($\wh A_{c}^{+}$) have
the highest confidence of being at risk, while the regions outside
the green boundary ($\wh A_{c}^{-}$) have the highest confidence
of not being at risk. Over repeated sampling, there is a probability of about $90\%$ that
the regions at risk include those within the red boundary and exclude
those outside the green boundary. 

The problem of finding confidence sets for spatial excursion sets,
sometimes also called exceedance regions or level sets, has been studied
in the past in two major contexts that substantially differ from the problem under consideration here. In the geostatistics literature,
the target function is itself a Gaussian field. In consequence, the excursion and the contour sets are random themselves. The data in this setting is a partial realization of the field, that is, the values of a realization of the field at relatively few spatial locations. This severe limitation of available information is compensated by assuming that the covariance structure of the field is known.
  This problem has been addressed from a frequentist perspective
in terms of confidence regions for level contours \citep{Lindgren1995,Wameling2003,French2014}
and for excursion sets  \citep{French2013}. 
Incidentally, our techniques share some similarities with \cite{French2014}, although we will show that distinguishing between level contours and excursion sets is important.
In a Bayesian setting for latent Gaussian models, \cite{Bolin2014} address uncertainty in both, contours and excursion sets.

The  second setting in which the problem has received attention is non-parametric density estimation and regression.  Here, the target function is a
probability density or regression function, estimated from realizations of a random variable with values in $\RR^q$ for some $q$. While the estimation of both level sets and contours have been well studied
\citep{Tsybakov1997,Cavalier1997,Cuevas2006,Willett2007,Singh2009,Rigollet2009}, there is less literature on confidence statements.  
\cite{Mason2009} showed asymptotic normality of plug-in level set estimates with respect to the measure of symmetric set difference. \cite{Mammen2013} proposed a bootstrapping scheme to obtain confidence sets analogous to our CoPE sets from vector-valued samples.

The problem of finding the threshold for our CoPE sets involves computation of the tail probability of the supremum of a limiting Gaussian random field. In \cite{French2014} this computation was done by Monte Carlo simulation assuming that the covariance structure of the field is known. More generally for unknown covariance function, as we attempt here, this problem was solved elegantly by \cite{Taylor2007} using the Gaussian kinematic formula. However, this method requires that the observations themselves be Gaussian and requires the field to be differentiable. The multiplier bootstrap allows us to avoid both these assumptions while being extremely fast to compute.
We  compare the finite sample performance of the Gaussian kinematic formula method and the multiplier bootstrap in a simulation.

All computations in this paper were performed using \textup{R} \citep{R2014}. All required functions for computation and visualization of CoPE sets and in particular an implementation of the Algorithm \ref{Algorithm} are available in the \textup{R}-package \textup{cope}  \citep{cope2015}.

\begin{figure}
\begin{centering}
\subfloat[CoPE sets for increase of mean summer temperature.]{\includegraphics[width=0.45\textwidth]{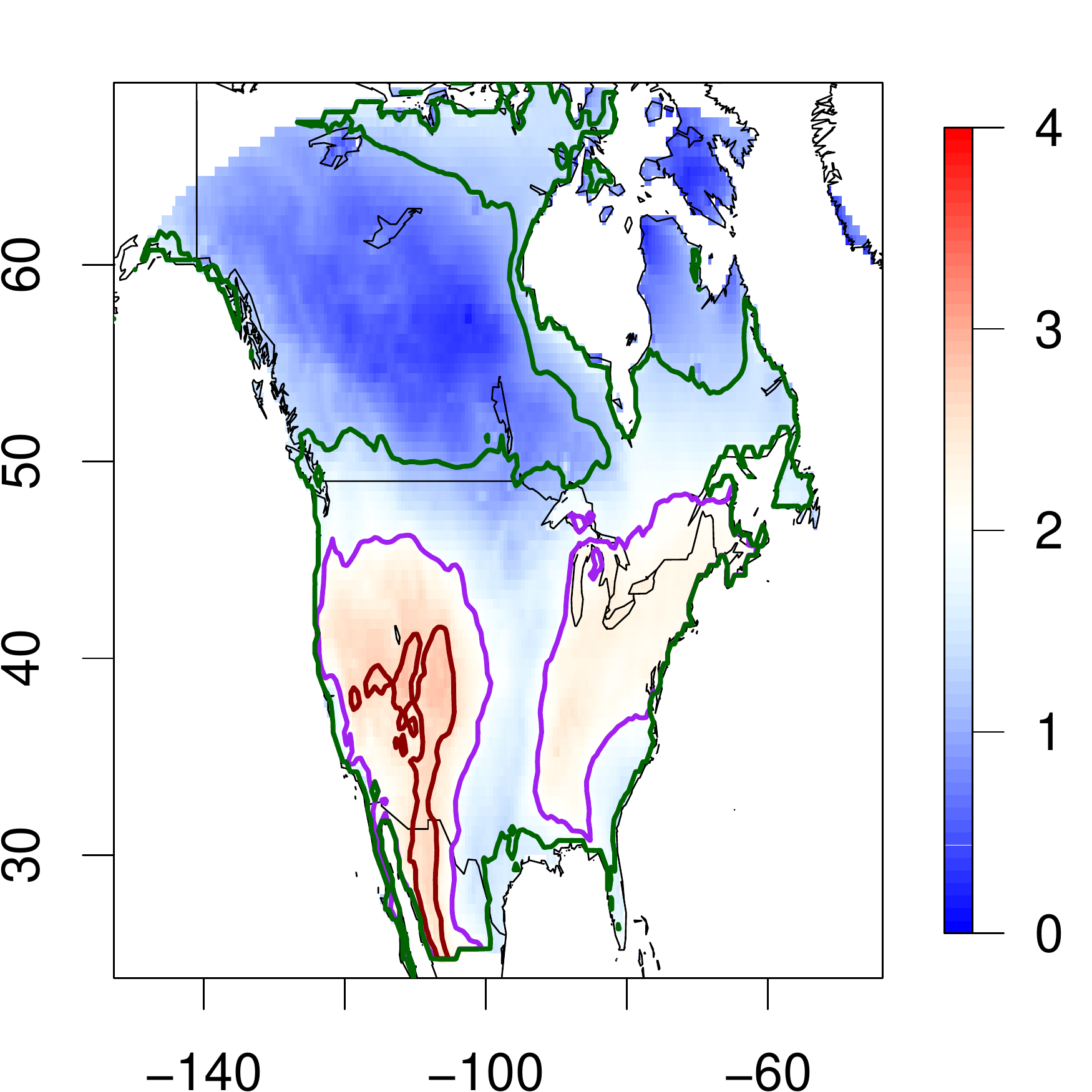}}\hfill{}\subfloat[CoPE sets for increase of mean winter temperature.]{\includegraphics[width=0.45\textwidth]{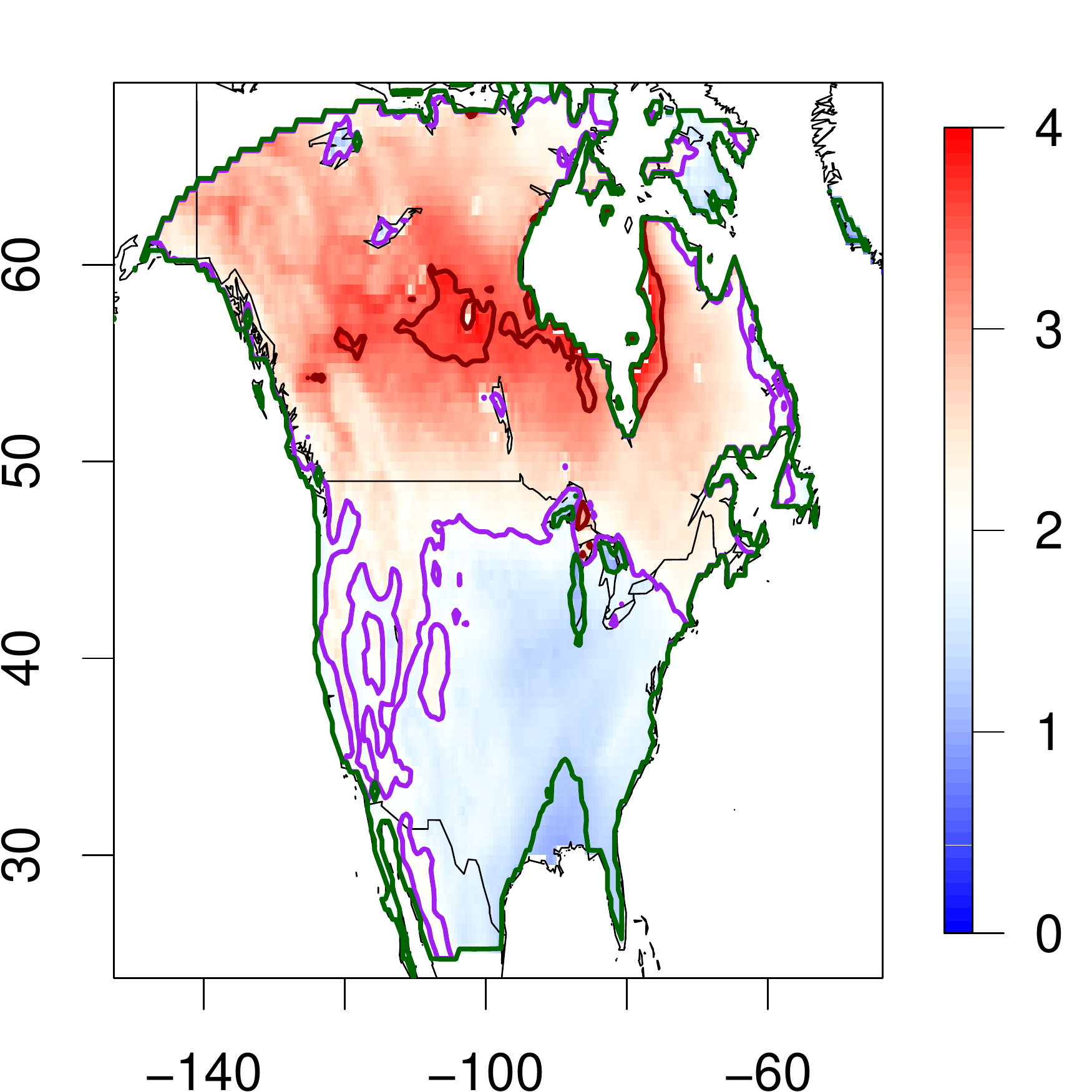}}
\par\end{centering}

\protect\caption{Output of our method for the increase $b_{1}(\protect\s)$ of the
mean summer (June-August) and winter (December-February) temperatures. Shown are heat maps of the estimator
$\protect\wh b_{1}(s)$. The uncertainty in the excursion set estimate
$\protect\wh A_{c}=A_{c}(\protect\wh b_{1})$ (purple boundary) for
$c=2^{\circ}C$ is captured by the CoPE sets $\protect\wh A_{c}^{+}$
(red boundary) and $\protect\wh A_{c}^{-}$ (green boundary). The
threshold was obtained according to Theorem \ref{thm:inclusion} to
guarantee inclusion $\protect\wh A_{c}^{+}\subset A_{c}\subset\protect\wh A_{c}^{-}$
with confidence $1-\alpha=0.9$.
The horizontal and vertical axes are indexed in degrees longitude and latitude, respectively.}
\label{data_res} 
\end{figure}

\subsubsection*{Outline of the paper}

In Section \ref{sec:Error-control} we propose a thresholding scheme to obtain CoPE
sets $\wh A_{c}^{\pm}$ as in (\ref{eq:inclusion}) from an estimator
$\wh{\mu}$ of $\mu$, only requiring continuity  of $\mu$  and, most importantly,
that $\wh{\mu}$ is asymptotically Gaussian. We show that the
asymptotic coverage probability is equal to
the tail probability of the limiting Gaussian field on the boundary
$\partial A_{c}$ of the excursion set $A_{c}$. 

Section \ref{sec:CoPEandFARE_for_glm} is devoted to presenting results
and algorithms for the construction of CoPE sets when the target function
is the parameter function in a general linear model. First, in Section
\ref{sub:CLTs-for-parameters} we derive central limit theorems for
these quantities. Then, in Section \ref{sub:CoPE-sets:-Approximating}
we show how to obtain the threshold for the construction of CoPE sets
by an efficient multiplier bootstrap. We compare it with a method
for Gaussian smooth noise based on \citet{Taylor2007}.
Section \ref{sub:Algorithm} combines the previous results in a concise
algorithm for the construction of CoPE sets.

Section \ref{sec:Simulations} shows a toy example to investigate
the non-asymptotic performance of CoPE sets and finally, in Section
\ref{sec:Application-to-climate}, we apply our methods for a general
linear model to the climate data. All proofs are in the appendix.

\section{\label{sec:Error-control}\label{sub:Controlling-the-coverage}Error
control for excursion sets - CoPE sets}

The domain $S\subset\RR^{N}$ on which all our functions and processes
are defined, is assumed to be a compact but not necessarily connected
subset of Euclidean space. 
We call the topological boundary $\partial A_c$ of the excursion set $A_c$ the \textit{contour of $\mu$ at the level $c$}. 
\begin{assumption}
\label{assu:CLT_and_mu_muhat}We assume that 
\begin{enumerate}
	\item the target function $\mu$ is continuous and the level set $\{\s:\mu(\s)=c\}$ is equal to $\dAc$.
	\item the estimator $\wh{\mu}_{n}(\s)$
	is continuous in $\s$ (for all $n\in\NN$).
	\item there is a sequence of numbers $\tau_{n}$ and a continuous function
	$\sigma:S\ra\RR^{+}$ such that
	\begin{equation}
	\frac{\wh{\mu}_{n}(\s)-\mu(\s)}{\tau_{n}\sigma(\s)}\ra G(\s),\label{eq:CLT-1}
	\end{equation}
	weakly in $C(S)$. Here, $G$ is a Gaussian field on $S$ with mean
	zero, unit variance and continuous sample paths with probability one.
\end{enumerate}
\end{assumption}
We will obtain nested estimates by thresholding the surface $\wh{\mu}_{n}(\s)$
as follows: 
\begin{equation}
\hat{A}_{c}^{+}  :=\wh A_{c}(+a):=\left\{ \s:\frac{\hat{\mu}_{n}(\s)-c}{\tau_{n}\sigma(\s)}\ge+a\right\} ,\quad
\hat{A}_{c}^{-}  :=\wh A_{c}(-a):=\left\{ \s:\frac{\hat{\mu}_{n}(\s)-c}{\tau_{n}\sigma(\s)}\ge-a\right\} ,
\label{eq:upper-lower-contours}
\end{equation}
where $a^{+}$ and $a^{-}$ are appropriate non-negative constants
to be determined. Note that in this notation $\wh A_{c}=\wh A_{c}(0)$
and $\wh A^{\pm}$ are themselves excursion sets. Moreover, for any
choice of $a\geq0$ we have the inclusions $\wh A_{c}^{+}\subset\wh A_{c}\subset\wh A_{c}^{-}$,
and hence the estimates obtained via (\ref{eq:upper-lower-contours})
are in fact nested. The function used to define the excursion sets
is similar to the test statistic used in \citet{French2013}.

The following main result shows how the constant $a$ in \eqref{eq:upper-lower-contours}
can be chosen such that $\wh A_{c}^{+}\subset A_{c}\subset\wh A_{c}^{-}$
with a predefined probability.
\begin{thm}
\label{thm:inclusion}If the Assumptions \ref{assu:CLT_and_mu_muhat}
hold, then
\[
\lim_{n\ra\infty}P\left[\wh A_{c}^{+}\subset A_{c}\subset\wh A_{c}^{-}\right]=P\left[\sup_{\partial A_{c}}|G(s)|\leq a\right].
\]
\end{thm}
A direct consequence of Theorem \ref{thm:inclusion} is that $\hat{A}_{c}^{+}\subset A_{c}\subset\hat{A}_{c}^{-}$
with asymptotic probability at least $1-\alpha$ if we choose $a$ such that
$P\left[\sup_{\s\in\partial A_{c}}|G(\s)|\geq a\right]\leq\alpha$.
The determination of $a$ poses a computational challenge
since the distribution of the supremum of $|G(\s)|$ and the set $\partial A_{c}$
are unknown. In Section \ref{sub:CoPE-sets:-Approximating} we propose
an easy and fast way to approximate this distribution by a multiplier
bootstrap.

As mentioned in the Introduction,  confidence sets for the excursion set $A_c$ yield confidence sets for the contour $\partial A_c$. More precisely, we have the following
\begin{cor}\label{cor:inclusion_contour}
	Under the assumptions of Theorem \ref{thm:inclusion}, we have 
	\[
	\lim_{n\ra\infty} P\left[\partial A_c\subset \cl\left(\wh{A}_c^-\setminus\wh{A}_c ^+\right)\right]
	=P\left[\sup_{\partial A_{c}}|G(s)|\leq a\right],
	\]
	where $\cl$ denotes the topological closure.
\end{cor}
Note that, conversely, confidence sets for the contour do not automatically give confidence sets for the excursion set. In Figure \ref{fig:contourvsexcursion} we show a simple schematic example of a pair of nested sets $\wh{A}_c^\pm$ for which $\partial A_c\subset \wh{A}_c^-\setminus\wh{A}_c$ holds but  $\wh A_{c}^{+}\subset A_{c}\subset\wh A_{c}^{-}$ does not. 

\begin{wrapfigure}{R}{0.5\textwidth}
	\begin{center}
		\includegraphics[width=0.48\textwidth]{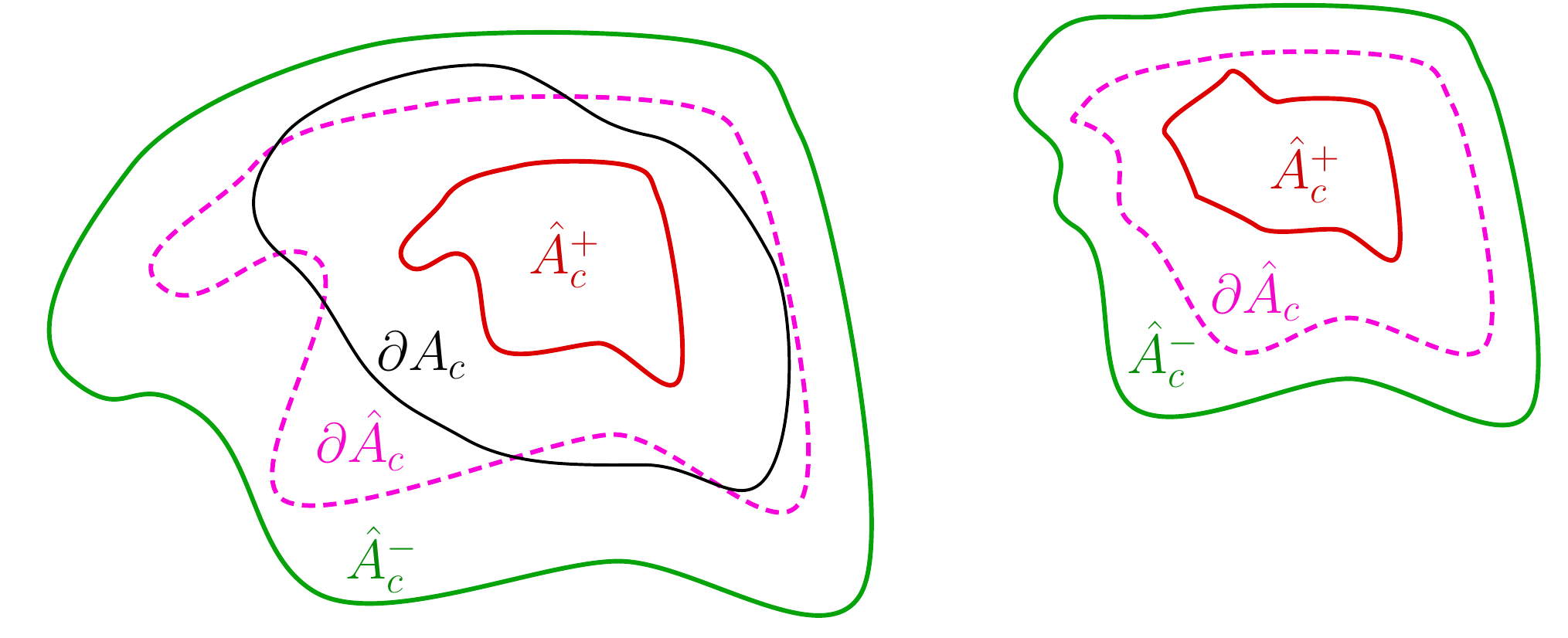}
	\end{center}
	\caption{A simple example of nested sets $\wh{A}_c^\pm$ that bound the contour $\dAc$ but do not satisfy the inclusion $\wh A_{c}^{+}\subset A_{c}\subset\wh A_{c}^{-}$} \label{fig:contourvsexcursion}
\end{wrapfigure}

In fact, excluding these cases is the more laborious part of the proof of Theorem \ref{thm:inclusion}. The key is to divide the region $S$ into a close-range zone where $\mu(\s)$ is close to $c$ and a long-range zone. More precisely, the close-range zone is given by the inflated boundary $A_{c}^{\eta}=\{\s\in S:c-\eta\sigma(\s)\leq\mu(\s)\leq c+\eta\sigma(\s)\}$. Then, the strategy of the proof is to let the parameter $\eta$ go to zero at an appropriate rate as $n\ra\infty$ such that, eventually, the probability of a part of $\wh{A}_c^+$ falsely appearing in the long-range zone $S\setminus A_{c}^{\eta}$ (as shown in Figure \ref{fig:contourvsexcursion}) vanishes. The probability of making an error remains in the close-range zone, and is asymptotically given by $P\left[\sup_{\partial A_{c}}|G(s)|> a\right]$.

We want to emphasize that Theorem 1 and its Corollary are valid for any
estimator $\hat{\mu}_n$ satisfying Assumption 1. Thus, they hold generally
whether the estimator is based on an increasing number of repeated observations
(like in our data example) or an increasing number of sampling spatial points
(like in the spatial statistics and nonparametric regression problems). However,
for concreteness, we focus on the former situation, which we develop in detail
in the following section.

\section{\label{sec:CoPEandFARE_for_glm}CoPE sets for general linear models}

For concreteness and application to the climate data, we here present
how CoPE sets are obtained, in theory and in practice, when the target function is a parameter
function in a general linear model and $n$ is the number of repeated
observations.

\subsection{\label{sub:CLTs-for-parameters}Asymptotic coverage probabilities}

We begin by proving an analog of Theorem \ref{thm:inclusion} for the parameters of a
general linear model. The most difficult part is to establish a Central Limit Theorem as in \eqref{eq:CLT-1}. This will require conditions on the error field as well as on the design. We consider the model
\begin{equation}
\Y(\s)=\X\b(\s)+\e(\s),\quad\s\in S\subset\RR^{N}\label{eq:glm}
\end{equation}
where $\Y(\s)$ is an $n\times1$ vector of observations, $\X$ is
a known $n\times p$ design matrix, $\b(\s)=(b_{1}(\s),\dots,b_{p}(\s))$
is an unknown $p\times1$ vector of parameters and $\e(\s)=\left(\ve_{1}(\s),\dots,\ve_{n}(\s)\right)$
with $\ve_{1},\dots,\ve_{n}\stackrel{i.i.d.}{\sim}\ve$ an unknown
stochastic process. Results of the kind presented here are well-known
(see e.g. \citet{Eicker1963}). We show and prove versions
tailored for our specific purpose for coherence and convenience. 

The least squares regression estimator for $\b(\s)$ in the model
(\ref{eq:glm}) is 
\[
\widehat{\b}(\s)=\left(\X^{T}\X\right)^{-1}\X^{T}\Y(\s).
\]
In the notation of Section \eqref{sec:Error-control}, the target
function $\mu$ is now one of the parameter functions of the model
\eqref{eq:glm}, $b_{1}$, say. Of course, the choice of $b_1$ is arbitrary and
any other coefficient of $\b$ may be considered, with the obvious modifications
of the
assumptions and theorems. Naturally, $\wh b_{1}(\s)$ now
plays the role of the estimator $\wh{\mu}_{n}(\s)$.

Further, define $\sigma:S\ra\RR_{\geq0}$ via $\sigma^{2}(\s)=\var\left[\ve(\s)\right]$
and the correlation function $\cv:S\times S\ra(-1,1)$ by 
\[
\cv(\s_{1},\s_{2})=\frac{\cov\left[\ve(\s_{1}),\ve(\s_{2})\right]}{\sigma(\s_{1})\sigma(\s_{2})}.
\]
Recall that for $1\leq p\leq\infty$, the $p$-norm $||A||_{p}$ of
a matrix $A$ is defined to be $||A||_{p}=\sup_{||x||_{p}=1}||Ax||_{p}$.
Hence, by definition $||Ax||_{p}\leq||A||_{p}||x||_{p}$ for all $x$.
In the special case $p=\infty$ the matrix norm $||A||_{\infty}$
is the maximum absolute row sum of $A$, i.e. 
\[
||A=(a_{ij})_{1\leq i\leq m,1\leq j\leq n}||_{\infty}=\max_{1\leq i\leq m}\sum_{j=1}^{n}|a_{ij}|.
\]

\begin{defn}
	\begin{enumerate}[(a)]
		\item For vectors $\s,\t\in\RR^{N}$ define the \emph{block} $(\s,\t]=(s_{1},t_{1}]\times\dots\times(s_{N},t_{N}]\subset\RR^{N}$
		and for a stochastic process $\ve(\s)$ with index set containing
		$(\s,\t]$ define the \emph{increment} of $\ve(\s)$ around $(\s,\t]$
		(cf. \citet{Bickel1971}) as 
		\[
		\ve\left((\s,\t]\right)=\sum_{\kappa_{1}=0,1}\cdots\sum_{\kappa_{N}=0,1}\left(-1\right)^{N-\sum_{j}\kappa_{j}}\ve\left(s_{1}+\kappa_{1}(t_{1}-s_{1}),\cdots,s_{N}+\kappa_{N}(t_{N}-s_{N})\right).
		\]
		\item We denote the Lebesgue measure of a set $A\subset S$ by $|A|$. For non-negative numbers $\delta$, $\gamma$, $\beta$ we say that
		the error field $\ve(\s)$ has the properties 
		\begin{itemize}
			\item \textsf{{N1-$\delta$}}, if $\sup_{\s\in S}\sigma(\s)^{-(2+\delta)}E|\ve(\s)|^{2+\delta}<\infty$; 
			\item \textsf{{N2-$(\gamma,\beta)$}}, if there exists a constant $C>0$ such that
			$E|(\sigma^{-1}\ve)(B)|^{2+\gamma}\leq C\left|B\right|^{1+\beta}$
			for all blocks $B\subset S$.
		\end{itemize}
	\end{enumerate}
\end{defn}
In dimension $N=1$ the definition of an increment yields $\ve((s_{1},t_{1}]=\ve(t_{1})-\ve(s_{1})$,
the usual increment. In Dimension $N=2$ we get 
$\ve((\s,\t])=\ve(t_{1},t_{2})-\ve(s_{1},t_{2})-\ve(t_{1},s_{2})+\ve(s_{1},s_{2}).$

\begin{assumption}
\label{ass:noise-assumption-1}Assume that 
\begin{enumerate}
\item the parameter functions $\b(\s) = (b_1(\s),\dots, b_p(\s))$ are continuous and the level set $\{\s:b_1(\s)=c\}$ is equal to $\partial A_c(b_1)$.
\item the noise field $\ve(\s)$ has continuous sample paths with probability
one. Moreover, a centered unit variance Gaussian field with correlation function $\cv(\s_1,\s_2)$ also has continuous sample paths with probability one.
\item the variance function $\sigma(\s)$ is continuous.
\item there exists a $\delta>0$ such that $\ve(\s)$ has the property \textsf{\textup{N1}}-$\delta$
and $n\left|\left|\X(\X^{T}\X)^{-\nicefrac{1}{2}}\right|\right|_{\infty}^{2+\delta}\ra0$
as $n\ra\infty$. 
\item there exist $\gamma\geq0$ and $\beta>0$ such that $\ve(\s)$ has
the property \textsf{\textup{N2}}-$(\gamma,\beta)$ and $\max_{n\in\NN}n\norm{\X(\X^{T}\X)^{-\nicefrac{1}{2}}}_{\infty}^{2+\gamma}<\infty$. 
\end{enumerate}
\end{assumption}
Part (a) and (b) of Assumption \ref{ass:noise-assumption-1} are tantamount to the first two conditions in Assumption \ref{assu:CLT_and_mu_muhat}. Parts (c), (d) and (e) will ensure that the parameter function $b_1$ enjoys a Central Limit Theorem. Note that the assumptions on the increments of the error field and on the design matrix $\X$ are coupled. 
The following Theorem \ref{thm:weak-convergence} gives convergence results and explains how we can obtain CoPE sets for $b_1$.
\begin{thm}
\label{thm:weak-convergence}Under Assumption \ref{ass:noise-assumption-1}
the following is true.
\begin{enumerate}[(a)]
	\item the weak convergence 
	\[
	\sqrt{\X^{T}\X}\left(\wh{\b}(\s)-\b(\s)\right) / \sigma(\s)\ra G^{\otimes p}(\s),
	\]
	holds, where $G^{\otimes p}$ is an $\RR^{p}$-valued mean zero,
	unit variance Gaussian random field with correlation function 
	\[
	\cov\left[G^{\otimes p}(\s_{1}),G^{\otimes p}(\s_{2})\right]=\cv(\s_{1},\s_{2})\mathbf{I}_{p}.
	\]
	\item if additionally  the top-left entry $\pi_{n}=\left[\left(\X^{T}\X\right)^{-1}\right]_{11}$
	of the matrix $\X^{T}\X$ is not zero and with $e_{1}^{T}=(1,0,\dots,0)$,
	the first standard basis vector, 
	\[
	\pi_{n}^{-\nicefrac{1}{2}}e_{1}^{T}(\X^{T}\X)^{-\nicefrac{1}{2}}\ra\mathbf{v}^{T}\in\RR^{p},\quad\mbox{as\ }n\ra\infty,
	\]
	then we have
	\[
	||\mathbf{v}||_{2}^{-1}\pi_{n}^{-\nicefrac{1}{2}}\sigma(\s)^{-1}(\wh b_{1}(\s)-b_{1}(\s))\ra G(\s),
	\]
	weakly, where $G$ is a mean zero, unit variance Gaussian field on
	$S$ with correlation function $\cov\left[G(\s_{1}),G(\s_{2})\right]=\cv(\s_{1},\s_{2})$. 
	
	\item under the additional assumptions of part (b), and if we define
          \begin{equation}\label{eq:cope_sets_glm}
        \hat{A}_{c}^{\pm}(b_1)  :=\left\{
          \s:\frac{\hat{b}_1(\s)-c}{||\mathbf{v}||_{2}\pi_{n}^{\nicefrac{1}{2}}\sigma(\s)}\ge
        \pm a\right\} ,
	\end{equation}
	then 
	\[
	\lim_{n\ra\infty}P\left[\wh A_{c}^{+}(b_1)\subset A_{c}(b_1)\subset\wh A_{c}^{-}(b_1)\right]=P\left[\sup_{\partial A_{c}(b_1)}|G(s)|\leq a\right].
	\]
\end{enumerate}
\end{thm}

\subsection{\label{sub:CoPE-sets:-Approximating}Approximating the tail probabilities
of $G$}

\subsubsection{\label{sub:multiplier_bootstrap}Multiplier bootstrap}

In order to obtain CoPE sets from Theorem \ref{thm:weak-convergence} we
need to know the tail distribution of the supremum of the limiting
(non-stationary) Gaussian field $G$. In applications, as for example
our climate data, the distribution of $G$ (and hence of its supremum)
is unknown, because it depends on the unknown (non-stationary) covariance
function. In our motivating application the only information we have
about $G$ is contained in the residuals $(R_{1}(\s),\dots,R_{n}(\s))=\R(\s)=\Y(s)-\X\wh{\b}(\s)$
of the linear regression.

A way of approximating the distribution of the limiting Gaussian field
$G$ in this situation is given by the multiplier or wild bootstrap
first introduced by \citet{Wu1986} and later studied by
\citet{Mammen1992,Mammen1993,Hardle1993}.
It is based on the following idea. Let $g_{1},\dots,g_{n}$ be i.i.d
standard Gaussian random variables independent of the data. Consider
the random field
\begin{equation}
\tilde{G}(\s)=n^{-\frac{1}{2}}\sum_{j=1}^{n}g_{j}R_{j}(\s).\label{eq:Mult_Boot}
\end{equation}
Then, conditional on the residuals $\left\{ R_{j}(\s)\right\} _{j=1}^{n}$,
the field $\tilde{G}(\s)$ is Gaussian and has covariance 
\begin{align}
\cov\left[\tilde{G}(\s_{1}),\tilde{G}(\s_{2})\right] & =\frac{1}{n}\sum_{i,j=1}^{n}R_{i}(\s_{1})R_{j}(\s_{2})\cov\left[g_{i},g_{j}\right]\nonumber \\
 & =\frac{1}{n}\sum_{j=1}^{n}R_{j}(\s_{1})R_{j}(\s_{2}),\label{eq:sample_cov}
\end{align}
the sample covariance. For large $n$, we expect the sample covariance to resemble the true covariance. The idea is to take the distribution of $\tilde{G}$  as an approximation of the distribution of $G$. In particular,
we can approximate $P\left[\sup_{\s\in\partial A_{c}}|G(\s)|\leq a\right]$, needed in Theorem \ref{thm:inclusion},
by $P\left[\sup_{\s\in\partial A_{c}}|\tilde{G}(\s)|\leq a|\left\{ R_{j}(\s)\right\} _{j=1}^{n}\right]$.
In practice, the latter can be efficiently computed by generating a large number $M$ of
i.i.d. copies $\tilde{G}_{1}(\s),\dots,\tilde{G}_{M}(\s)$ of $\tilde{G}(\s)$,
conditional on $\left\{ R_{j}(\s)\right\} _{j=1}^{n}$, via \eqref{eq:Mult_Boot}
and evaluating $M^{-1}\sum_{j=1}^{M}\one\left[\sup_{\s\in S}\left|\tilde{G}_{j}(\s)\right|\leq a\right]$.

{\bf Theoretical considerations:} The above approximation of the distribution of the supremum requires some justification because the sample covariance \eqref{eq:sample_cov} itself is not a good estimator of the true covariance function in our high dimensional setting, where the number of locations is much higher than the sample size $n$ (about ten thousand grid points vs. 58 field realizations in the climate data). However, the claim is about the distribution of the supremum of the process instead. For a discrete set of locations, the  distribution of the supremum $\sup_{\s\in\dAc}|G(\s)|$ is similar to the distribution of the maximum of a high-dimensional Gaussian random vector, recently considered by \cite{Chernozhukov2013}. Substantially extending the results of \cite{Mammen1993} in the high dimensional setting, \cite{Chernozhukov2013} show that the distribution of the maximum can be well approximated  by the Gaussian multiplier bootstrap using realizations of a not necessarily Gaussian random vector with the same covariance matrix. In this sense, the multiplier bootstrap is valid in our setting. This is confirmed by simulations in Section 4 below.

{\bf Computational considerations:} Besides these theoretical considerations, the multiplier bootstrap
is also computationally attractive. In comparison with the direct
simulation of the limiting field. While it is theoretically possible to
simulate a number of realizations of a non-stationary Gaussian field with a given
covariance as in \eqref{eq:sample_cov} and to obtain tail probabilities
from these, in practice this is  computationally infeasible. Assuming that all fields
are observed at $L$ locations in $S$ the direct method first requires
computing the $L\times L$ covariance matrix which is of complexity
$\mathcal{O}(nL^{2})$. Then, a Cholesky decomposition of the covariance
matrix must be computed in time $\mathcal{O}(L^{3})$. Finally, for
each realization a matrix-vector product with the triangular matrix
from the Cholesky decomposition is required and hence $N$ realizations
can be obtained in $\mathcal{O}(L^{2}N)$ time. This yields a total
complexity of $\mathcal{O}(L^{3}+L^{2}(n+N))$ which is  prohibitively
large (in our data $L=9051$). This problem has also been encountered
by \citet{Adler2012} in a setting where information about
the field is available through the true covariance function instead
of realizations of it.

In contrast, creating one multiplier bootstrap realizations of the
field $\tilde{G}$ requires computing a linear combination of $n$ vectors of dimension $L$ which  is of complexity $\mathcal{O}(nL)$. Hence, $N$ multiplier bootstrap realizations can be generated in $\mathcal{O}(nNL)$ time. Further,
note that the simulation of $N$ bootstrap realizations 
can be written as a matrix multiplication. Let $E$ be a $L\times n$
matrix such that each column is one residual $R_{j}$ and let $V$
be a $n\times N$ matrix with i.i.d. standard Gaussian entries. Then,
the columns of $EV$ correspond to realizations of $\tilde{G}$. This
makes the multiplier bootstrap very efficient because very fast implementations of  matrix
multiplication are available and it is an operation that can easily be parallelized

A comparison of the computational burden to existing methods is not easy, since all work that is known to the authors considers different settings. However, to put the above considerations in perspective, we briefly discuss computational requirements of the method proposed by \cite{French2014}. This method requires \cite[Sec. 3.7]{French2014} the computation of a kriging estimate in time $\mathcal{O}(n_l^3)$, where $n_l$ is the number of observed locations, and a Cholesky decomposition which is of complexity $\mathcal{O}((n_l+n_g)^3)$, where  $n_g$ is the number of grid locations (corresponding to $L$ in our notation). This increased complexity compared to the multiplier bootstrap is reflected in actual empirical computation times. According to his own experiments, the method of \cite{French2014} applied to $n=100$ observed locations and a grid size of $100\times100$, can be computed, on average, in just under five minutes. In contrast, the entire data analysis for the mean summer temperature, including linear regression at each grid point and computation of CoPE sets with the multiplier bootstrap,  with a grid size of $L\approx10,000$ and $n=58$ is performed in under five seconds on a machine comparable to the one used by \cite{French2014}.

\subsubsection{\label{sub:GKF}An alternative method for smooth fields}

If one is willing to assume that the limiting field is twice differentiable
and that the error field $\epsilon(s)$ in \eqref{eq:glm} is Gaussian itself then the Gaussian Kinematic
Formula (GKF) (e.g. \citet{Taylor2006,Taylor2005,Taylor2007})
offers another way to approximate tail probabilities. Our motivation
for presenting it here is twofold: First, it offers, under the additional
assumptions made above, an elegant and accurate way of computing tail probabilities
of Gaussian fields that does not require simulations and that has been successfully applied \citep{Taylor2007}.
Second, it shows that, at least for smooth fields, the tail probability
of the supremum is intrinsically low-dimensional. More precisely,
for a field on $\RR^{N}$ it is given (up to an exponential error
term) by $N+1$ numbers. This gives an additional justification for the ability of the multiplier bootstrap method to estimate the tail probability despite the high dimensionality of the field.

The GKF is based on two properties of smooth Gaussian fields. The first is that
for such fields the exceedance probability for high thresholds can
be approximated by the expected Euler characteristic. More precisely,
for any set $B\subset S$ with smooth boundary,
\begin{equation}
P\left[\sup_{\s\in B}G(\s)\geq a\right]=E\left[\chi(A_{a}(G)\cap B)\right]+\mathcal{O}(\exp(-a^{2}/2)),\label{eq:EEC-approx}
\end{equation}
where $\chi$ is the Euler characteristic (see e.g. \citet{Adler2007}).
The second property is that the expected Euler characteristic in \eqref{eq:EEC-approx}
has a closed formula. Indeed, with $\Lambda(\s)=\var\left[\dot{G}(\s)\right]$,
where $\dot{G}(\s)$ is the vector of partial derivatives of $G$,
we can write (cf. e.g. \citet{Taylor2006}) 
\begin{equation}
E\left[\chi(A_{a}(G)\cap B)\right]=\sum_{d=0}^{N}\mathcal{L}_{d}(B,\Lambda)\rho_{d}(a),\label{eq:EEC}
\end{equation}
with $\mathcal{L}_{d}$ the $d$-th order \emph{Lipschitz-Killing
curvature (LKC)} (see e.g. \citet{Taylor2006} for details
on these quantities) and known functions $\rho_{d}(a)$. We can use
this to obtain CoPE sets: If the Assumptions \eqref{assu:CLT_and_mu_muhat}
are satisfied then Theorem \eqref{thm:inclusion} implies in conjunction
with \eqref{eq:EEC-approx} and\eqref{eq:EEC} that 
\begin{align}
\lim_{n\ra\infty}P\left[\wh A_{c}^{+}\subset A_{c}\subset\wh A_{c}^{-}\right]&=  1-P\left[\sup_{\s\in\partial A_c}\left|G(\s)\right|\geq a\right]\nonumber\\
&=1-\sum_{d=0}^{N-1}\mathcal{L}_{d}(\partial A_{c},\Lambda)\rho_{d}(a)+\mathcal{O}(\exp(-a^{2}/2)).\label{eq:inclusion_EEC}
\end{align}
To use \eqref{eq:inclusion_EEC}, the problem amounts to estimating
the LKCs of the boundary $\partial A_{c}$. \citet{Taylor2007}
propose a method to estimate the LKCs based on a finite number of
realizations of $G$. Applying this method requires a triangulation
of the plug-in estimate $\text{\ensuremath{\partial}}\wh A_{c}$ of
the boundary $\partial A_{c}$.

While the triangulation is challenging, yet feasible, \citet{Taylor2007}
prove the validity of their method only when applied to realizations
of the Gaussian field $G$. In our application, however, we only have
realizations of a generally non-Gaussian field (the residuals in of
the linear model, cf. Section \ref{sec:Application-to-climate}) with
asymptotically the same covariance as $G$. For completeness, we compare this method to the multiplier bootstrap method in the simulations section below.

\subsection{\label{sub:Algorithm}Algorithm}

Combining the results of the previous sections we can give the exact
procedure for obtaining CoPE sets for the parameters of the linear
model $\Y(\s)=\X\b(\s)+\e(\s)$. 
\begin{lyxalgorithm}
\label{Algorithm}Given a design matrix $\X$ and observations $\Y(\s)$
following the linear model (\ref{eq:glm}). If Assumptions \ref{ass:noise-assumption-1}
hold, the following yields CoPE sets for $b_{1}(\s)$. 
\begin{enumerate}
\item Compute the estimate $\wh{\b}(\s)=\left(\X^{T}\X\right)^{-1}\X^{T}\Y(\s)$
and the corresponding residuals $\mathbf{R}(\s)=\Y(\s)-\X\wh{\b}(\s)$.
With the empirical variance $\wh{\sigma}^{2}(\s)= n^{-1}\sum_{j=1}^{n}R_{j}^{2}(\s)$
compute the normalized residuals $\tilde{\mathbf{R}}(\s)=\wh{\sigma}(\s)^{-1}\mathbf{R}(\s)$. 
\item Determine $a$ such that approximately $P\left[\sup_{\s\in S}\left|G(\s)\right|\geq a\right]\leq\alpha$.
For example, use the multiplier bootstrap procedure presented in Section
\ref{sub:multiplier_bootstrap} with the residuals $\left\{ \tilde{R}_{j}(\s)\right\} _{j=1}^{n}$
to generate i.i.d. copies of a Gaussian $\tilde{G}$ field with covariance
structure given by the sample covariance. With these, determine $a$
such that $P\left[\sup_{\s\in S}\left|\tilde{G}(\s)\right|\geq a\right]\leq\alpha$. 
\item Obtain the nested CoPE sets defined in equation \eqref{eq:cope_sets_glm} 
\end{enumerate}
\end{lyxalgorithm}

\section{\label{sec:Simulations}Simulations}

This section includes some artificial simulations to show that the
proposed methods provide approximately the right coverage in practical
non-asymptotic situations with non-smooth, non-stationary and non-Gaussian
noise. We will describe ways to obtain  error fields with these
properties. Our objective in the design of the error fields described
below is not to imitate the data but to introduce non-stationarity
and non-Gaussianity in a transparent and reproducible way, showing
the full potential of the method. In fact, the error field that we
encounter in the data (cf. Section \eqref{sec:Application-to-climate})
is better behaved as far as smoothness and stationarity
are concerned than the artificial fields we investigate here.

\subsection{\label{sub:Setup}Setup}

As a simple instance of the general linear model (\ref{eq:glm}) we
consider the signal plus noise model 
\[
y_{j}(\s)=\mu(\s)+\ve_{j}(\s),\quad j=1,\dots,n
\]
with non-stationary and non-Gaussian noise $\ve(\s)$ over a square
region $S$ of size $10\times 10$ consisting of $64\times64$
square pixels. We will consider three different noise fields and we
describe in the following how to obtain a realization of each.
\begin{description}
\item [Noise 1] In the upper half of $S$, each pixel is assigned the value of
a standard normal random variable, all of which are independent. In
the lower half, the pixels are grouped together in blocks of 4 by 4
pixels and each block is assigned the value of a standard normal,
again all independent(cf. Figure \ref{fig:noise1_presmooth}). Finally, the entire picture is convolved with a Gaussian kernel with bandwidth one and all values are multiplied by a scaling factor of $50$.
\item[Noise 2] Identical to \textsf{Noise 1} except the image is smoothed by a Laplace kernel with bandwidth one instead of a Gaussian and the scaling factor is $100$.
\item[Noise 3] Each pixel in the upper half is assigned the value of a Laplace distributed random variable with mean zero and variance two. In the lower half, pixels are assigned the values of independent Student $t$-distributed random variables with $10$ degrees of freedom. The entire picture is convolved with a Gaussian kernel of bandwidth one and multiplied by a scaling factor of $25$.
\end{description}

The noise fields \textsf{Noise 1-3} are intentionally designed to have non-homogeneous variance and scaling factors are chosen ad-hoc such that all three fields can be conveniently displayed on a common scale.
	
The signal $\mu$ is a linear combination of three Gaussians.
Figure \ref{fig:mu_noise} shows the signal $\mu$. Each one realization
of the three noise fields is shown in Figure \ref{fig:noise}.

We controlled the probability of coverage at the level $1-\alpha=0.9$
using Theorem \ref{thm:inclusion}. The estimator for $\mu$ here
is the mean $\wh{\mu}_{n}(\s)=n^{-1}\sum_{j=1}^{n}y_{j}(\s)$ and
the thresholds for the CoPE sets are obtained using Algorithm \ref{Algorithm}
with $||\mathbf{v}||_{2}^{-1}\pi_{n}^{-\nicefrac{1}{2}}=\sqrt{n}$.
The threshold $a$ was computed using the multiplier bootstrap procedure
proposed in Section \ref{sub:multiplier_bootstrap} using either the
true boundary $\partial A_{c}$ or the plug-in estimate $\partial\wh A_{c}$.
The results of our method using $\partial\wh A_{c}$ for each one
run with the three noise fields and sample sizes $n=60$, $n=120$
and $n=240$ are shown in Figure \ref{fig:res_toy}.

\begin{figure}
\centering{}
	\begin{tabular}{cc}
		\subfloat[\label{fig:mu_noise}The signal $\mu(\protect\s)$ with the true contour
		$\partial A_{c}(\mu)$ for $c=4/3$ in purple.]{\includegraphics[width=0.2\textwidth]{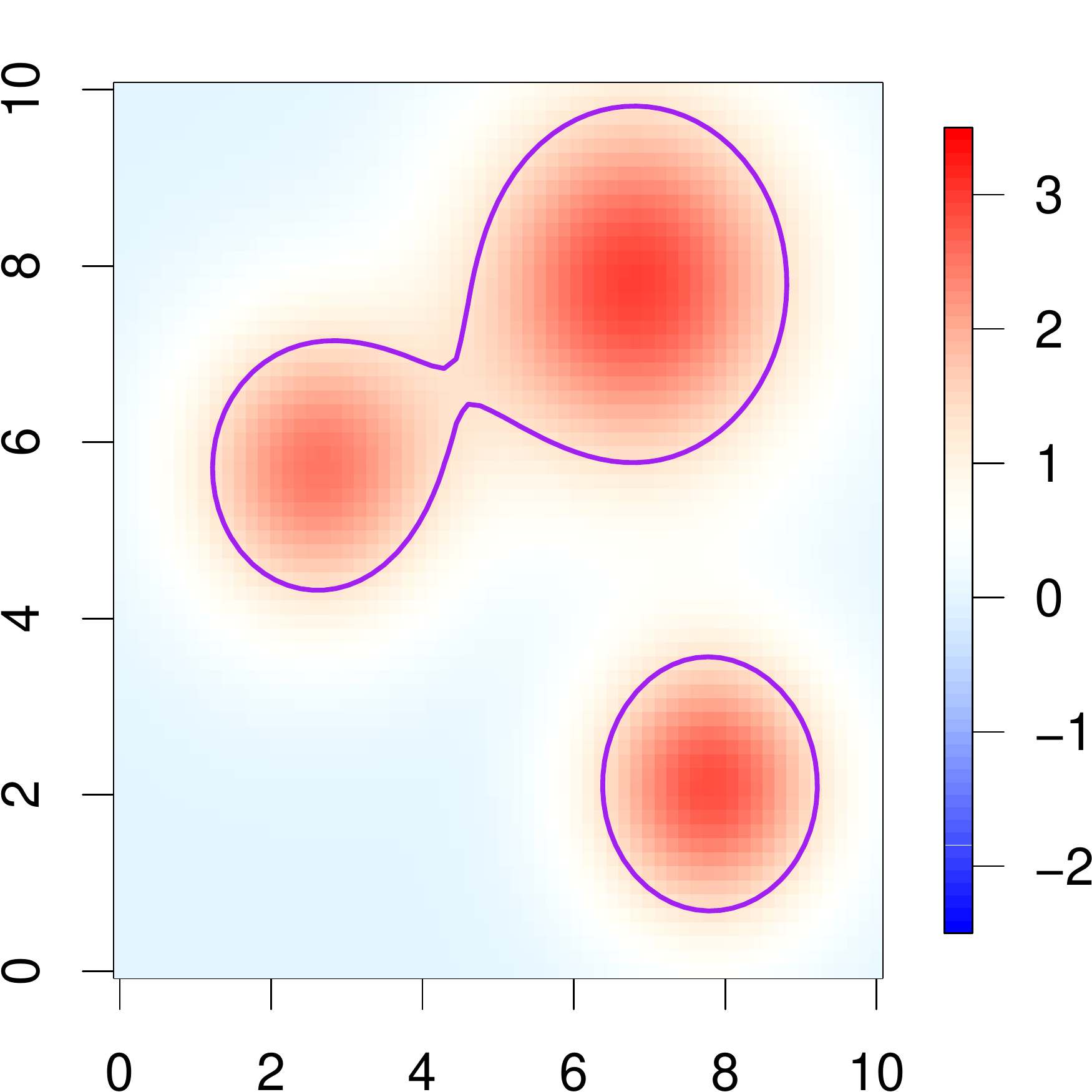}} &
		\subfloat[\label{fig:noise1_presmooth}A realization of \textsf{Noise 1} before smoothing with a Gaussian and scaling.]{\includegraphics[width=0.2\textwidth]{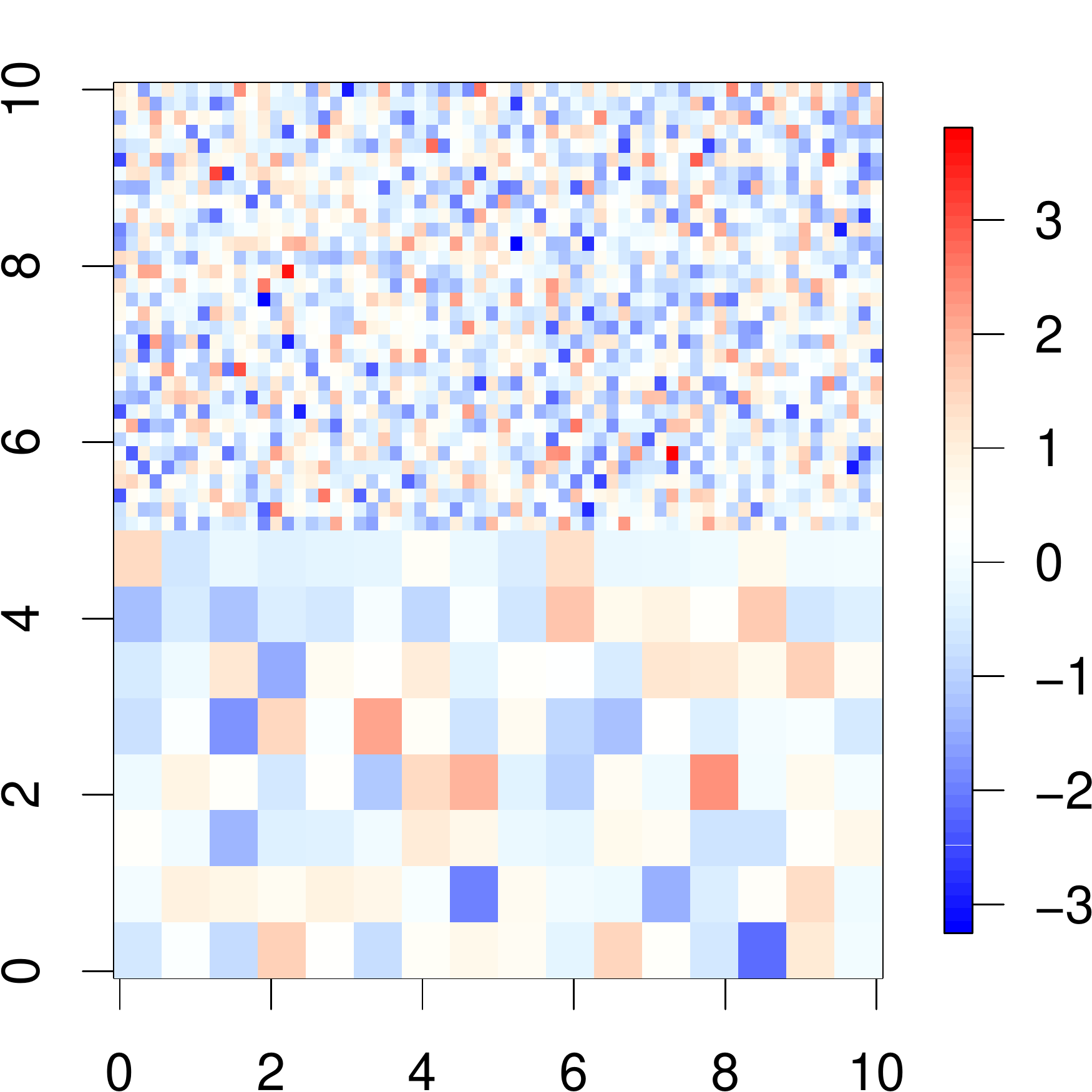}} 
	\end{tabular}
	\begin{tabular}{ccc}
		\subfloat[\textsf{Noise 1}. Nonstationary Gaussian noise smoothed with a Gaussian kernel.]{\includegraphics[width=0.2\textwidth]{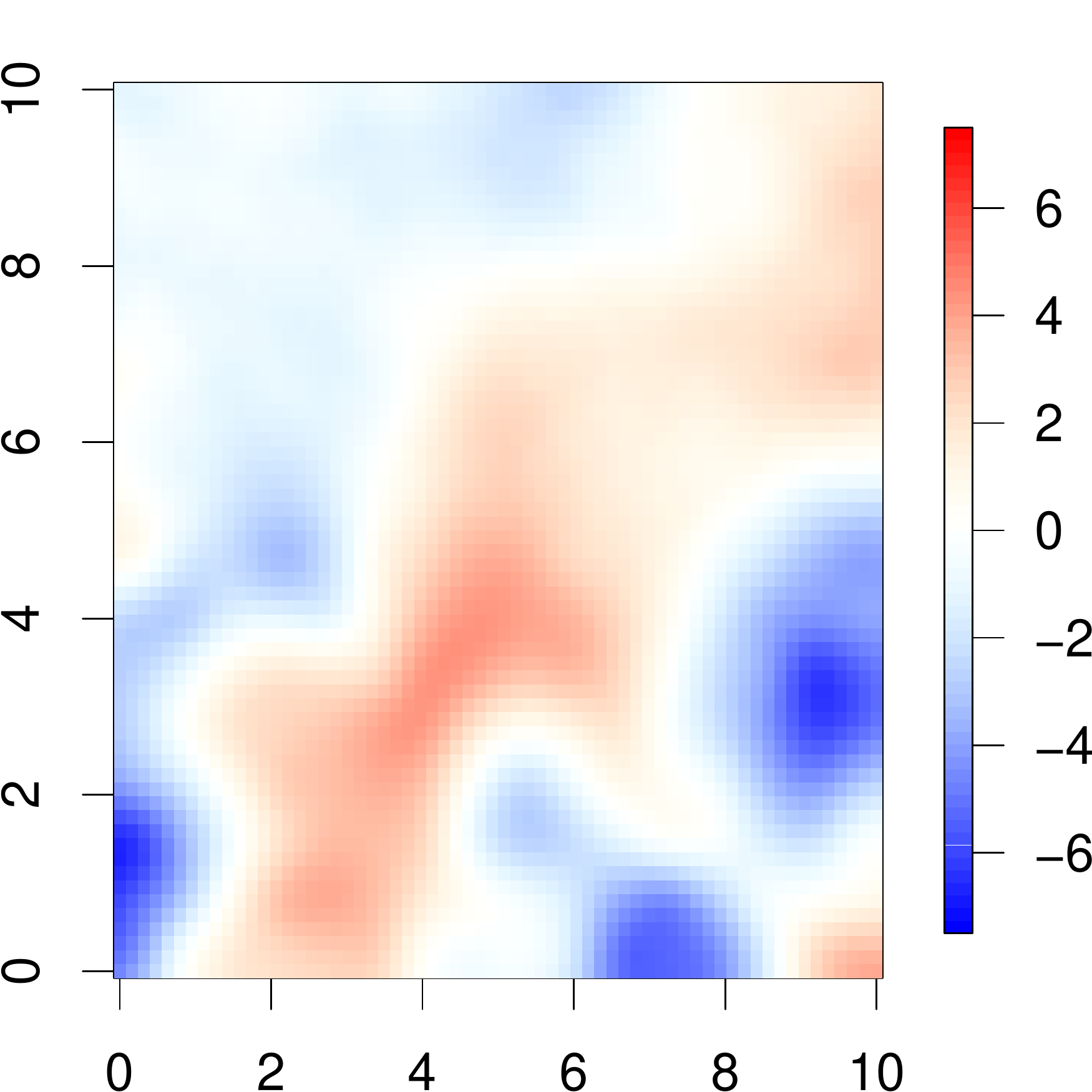}} & 
		\subfloat[\textsf{Noise 2}. Nonstationary Gaussian noise smoothed with a Laplace kernel.]{\includegraphics[width=0.2\textwidth]{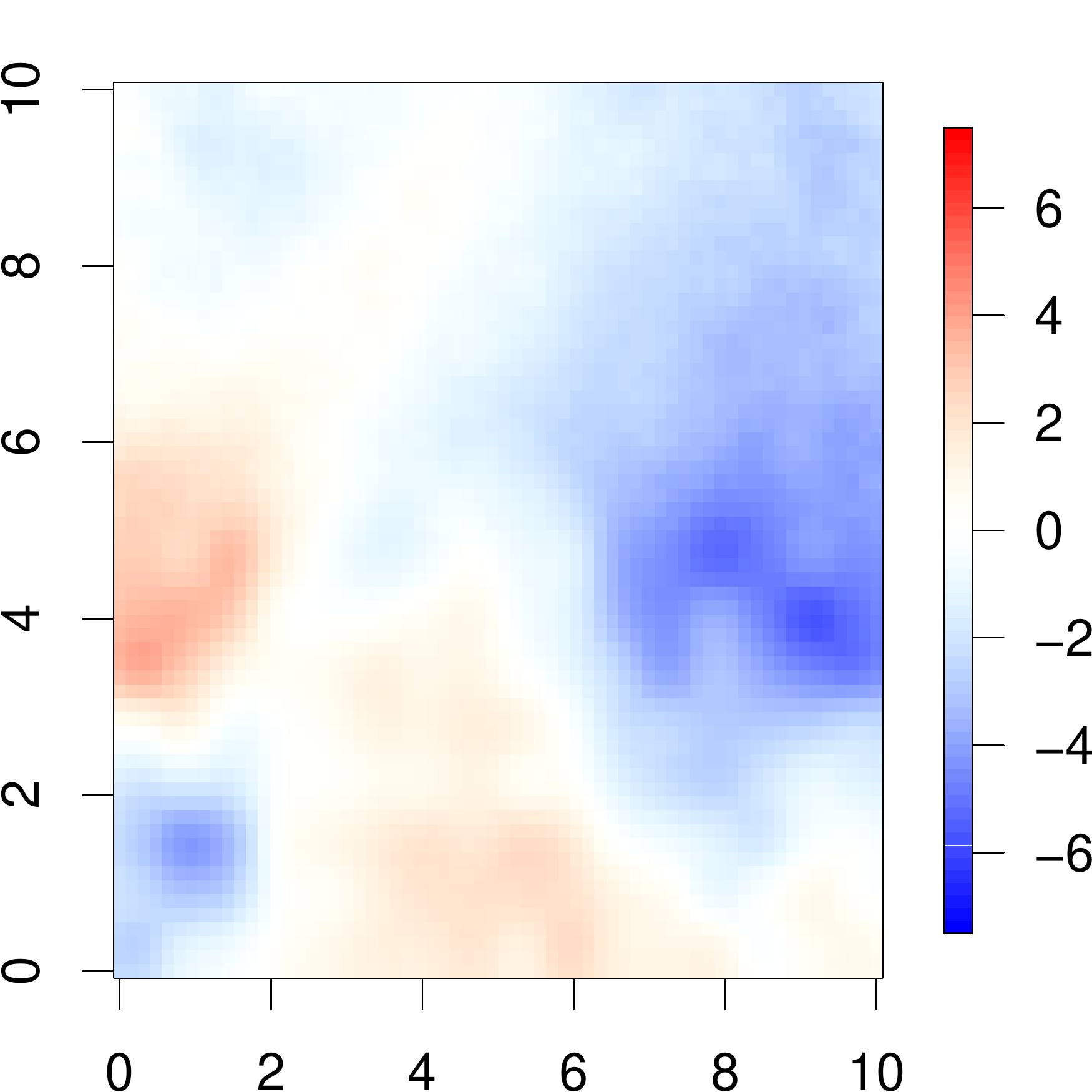}} & 
		\subfloat[\textsf{Noise 3}. Laplace and Student-$t$ distributed noise smoothed with a Gaussian
		kernel.]{\includegraphics[width=0.2\textwidth]{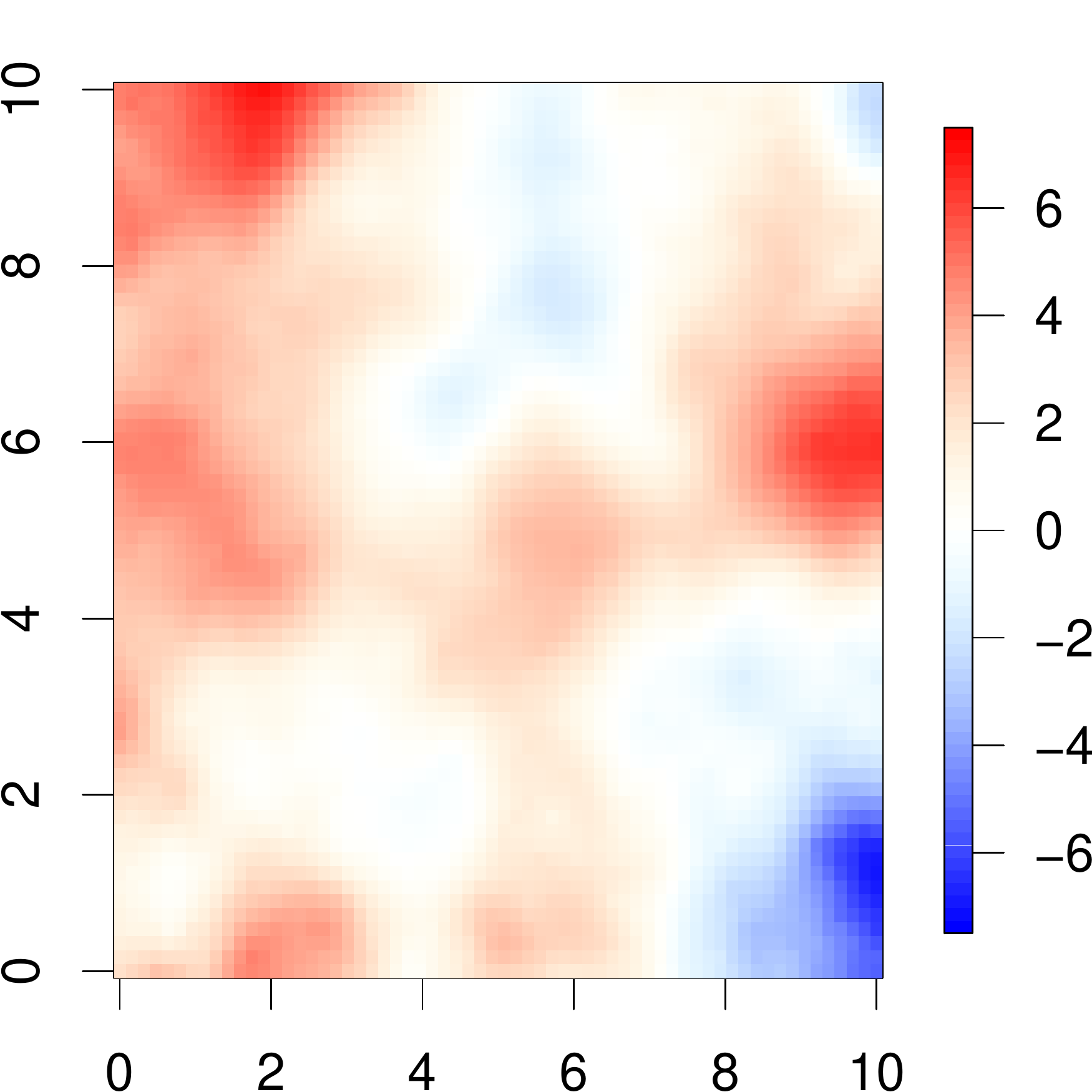}}
	\end{tabular}
\protect\caption{\label{fig:noise}The signal and noise in the toy example.}
\end{figure}

\begin{figure}
\begin{tabular}{cccc}
	& $n=60$ & $n=120$ & $n=240$\\
\begin{rotate}{90}\quad\textsf{Noise 1} \end{rotate}&\includegraphics[width=0.25\textwidth]{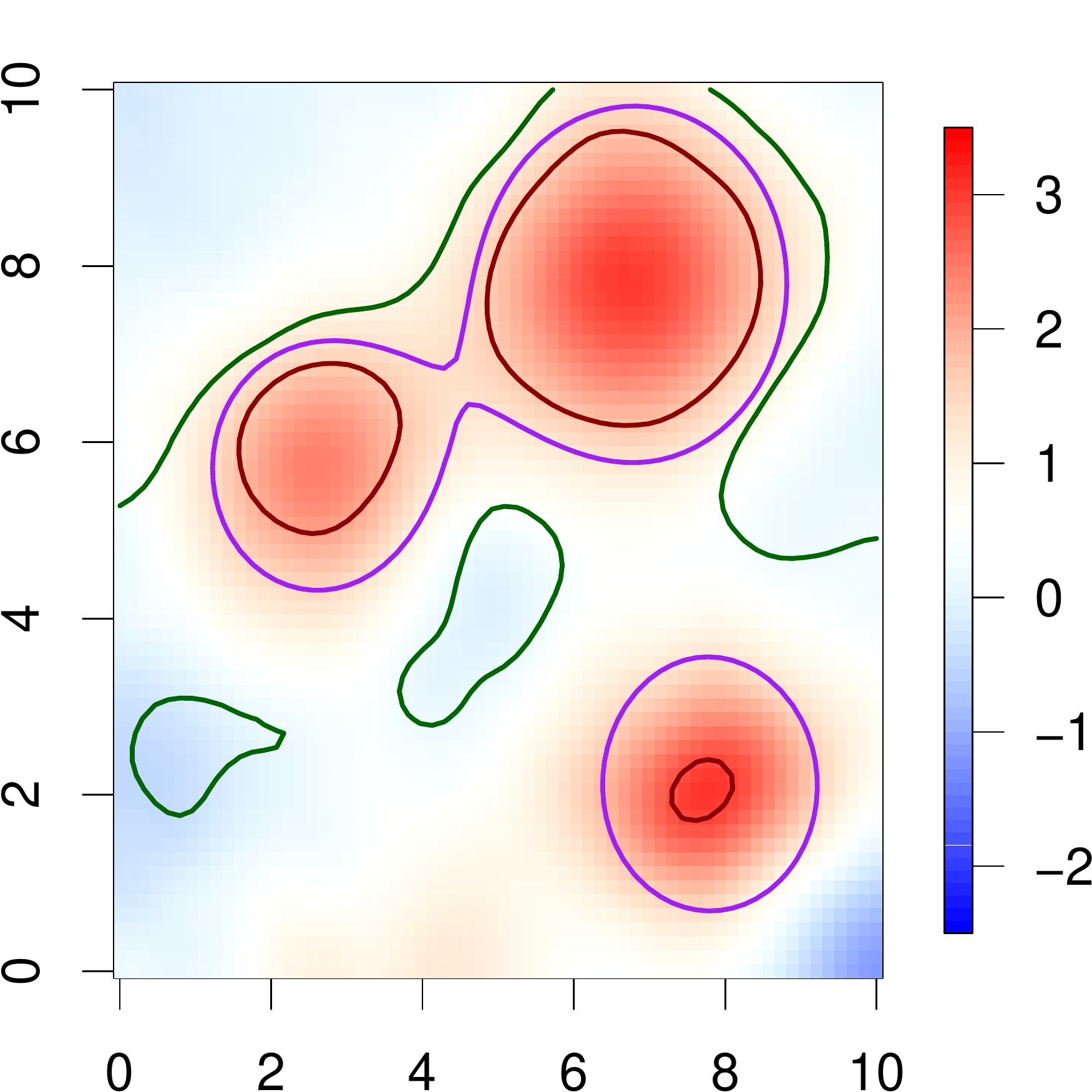}	& \includegraphics[width=0.25\textwidth]{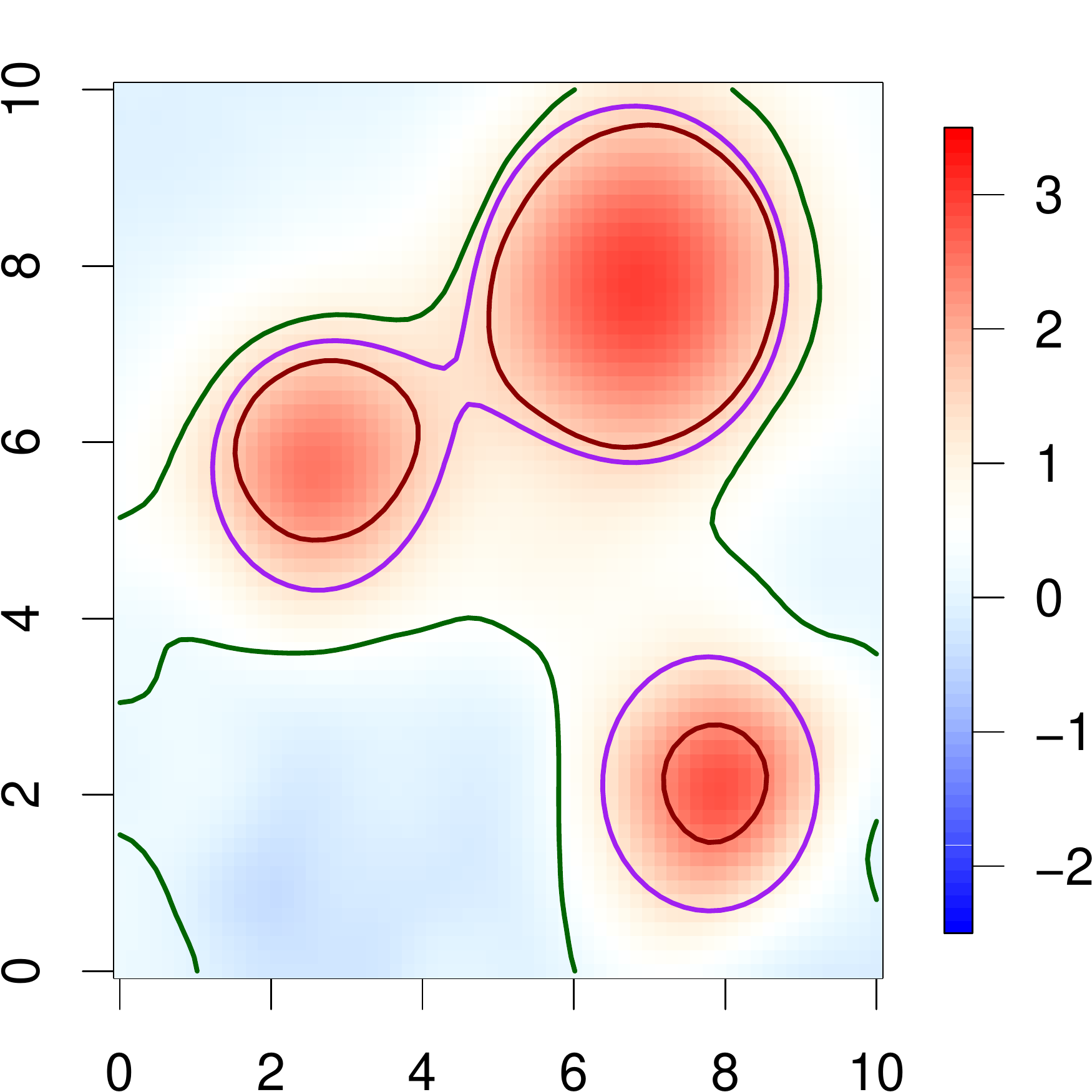} & \includegraphics[width=0.25\textwidth]{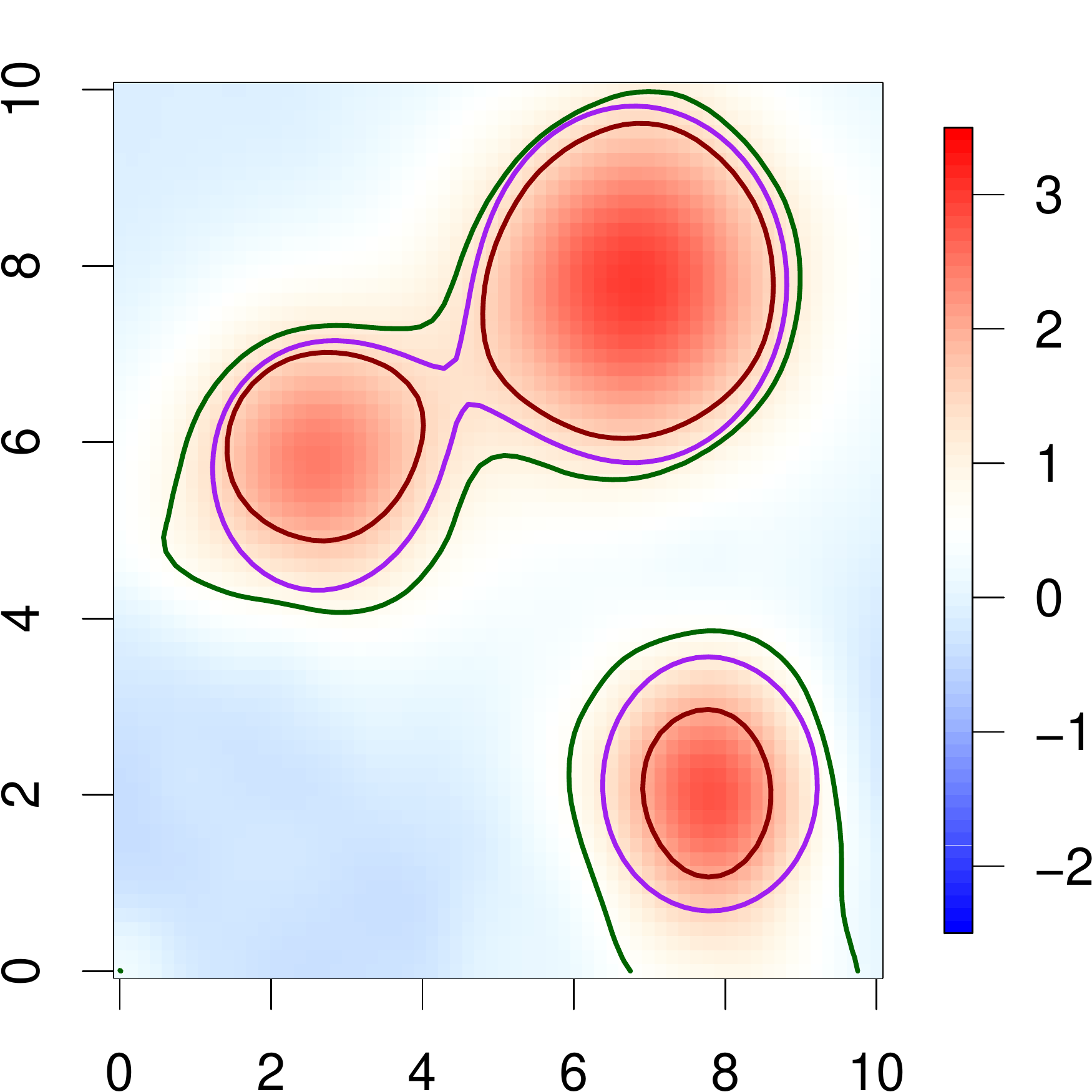} \\ 
\begin{rotate}{90}\quad \textsf{Noise 2} \end{rotate}&\includegraphics[width=0.25\textwidth]{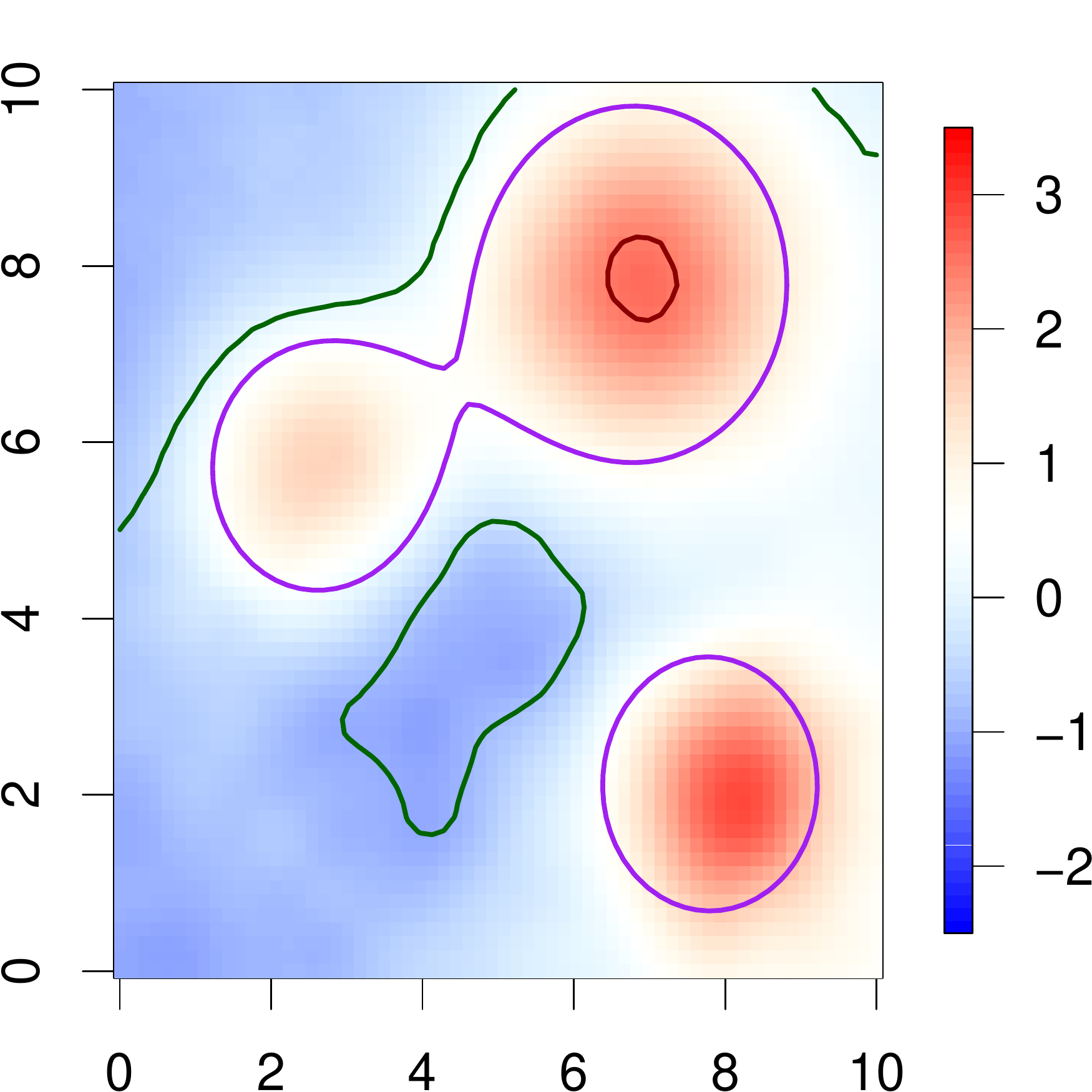}	& \includegraphics[width=0.25\textwidth]{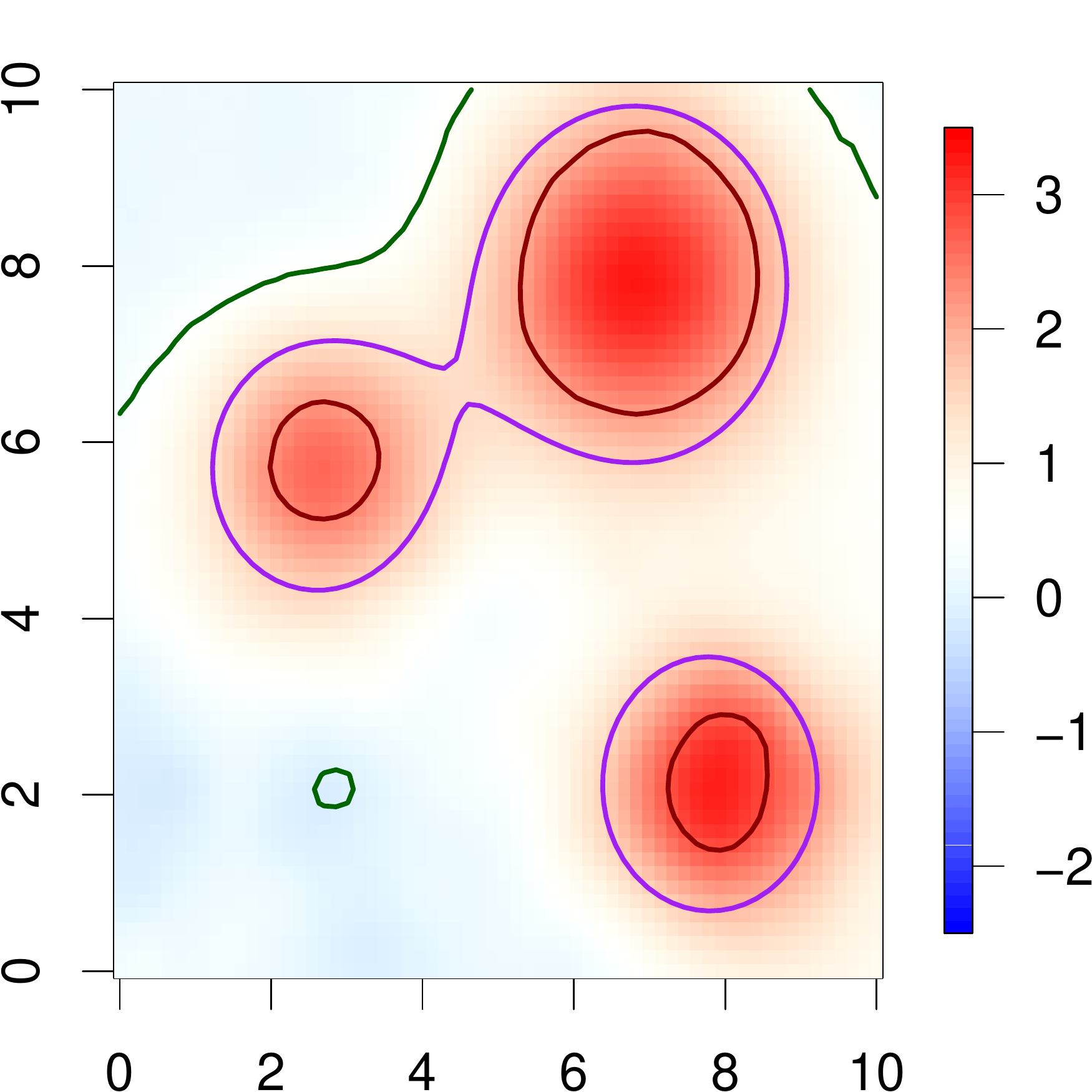} & \includegraphics[width=0.25\textwidth]{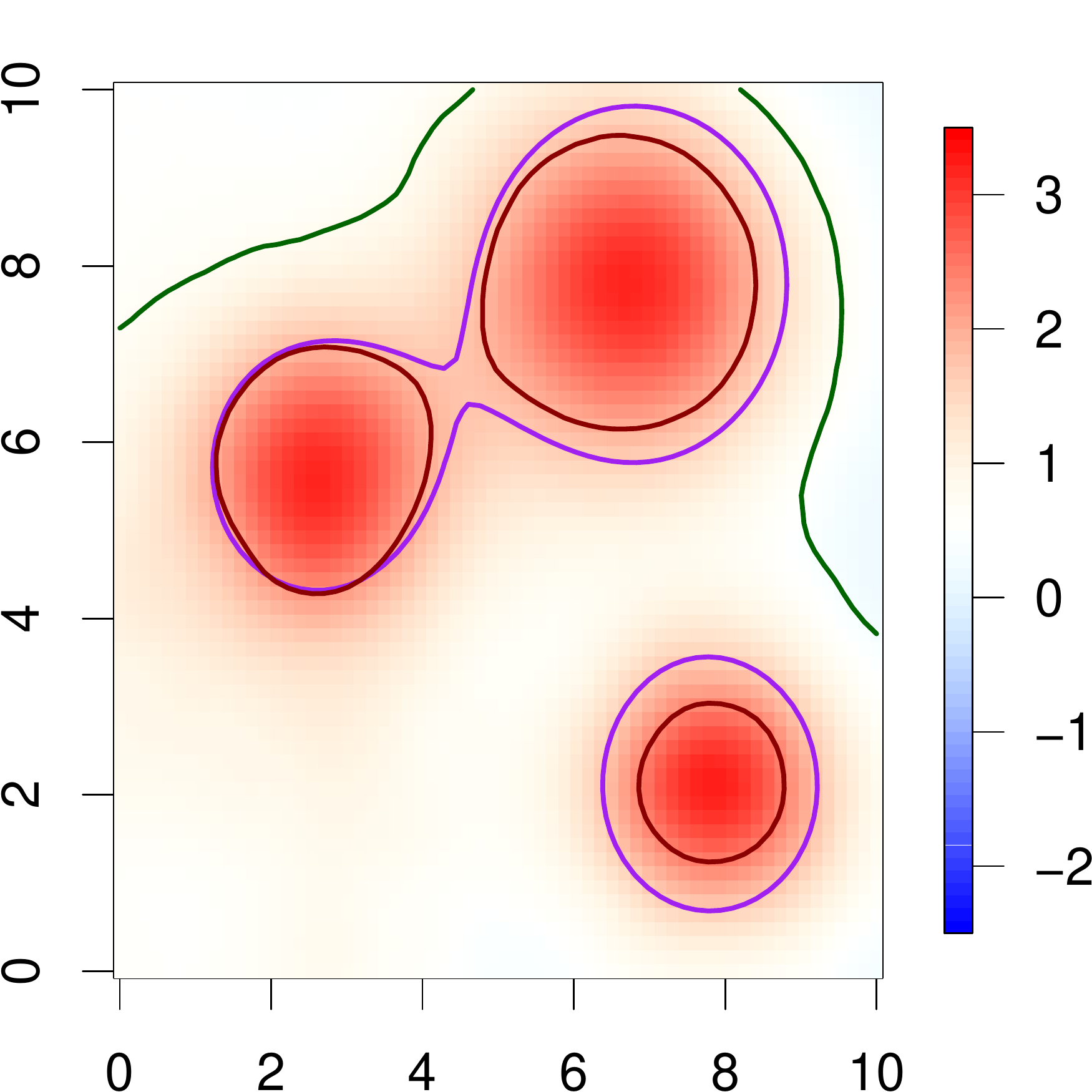} \\ 
\begin{rotate}{90}\quad \textsf{Noise 3} \end{rotate} &\includegraphics[width=0.25\textwidth]{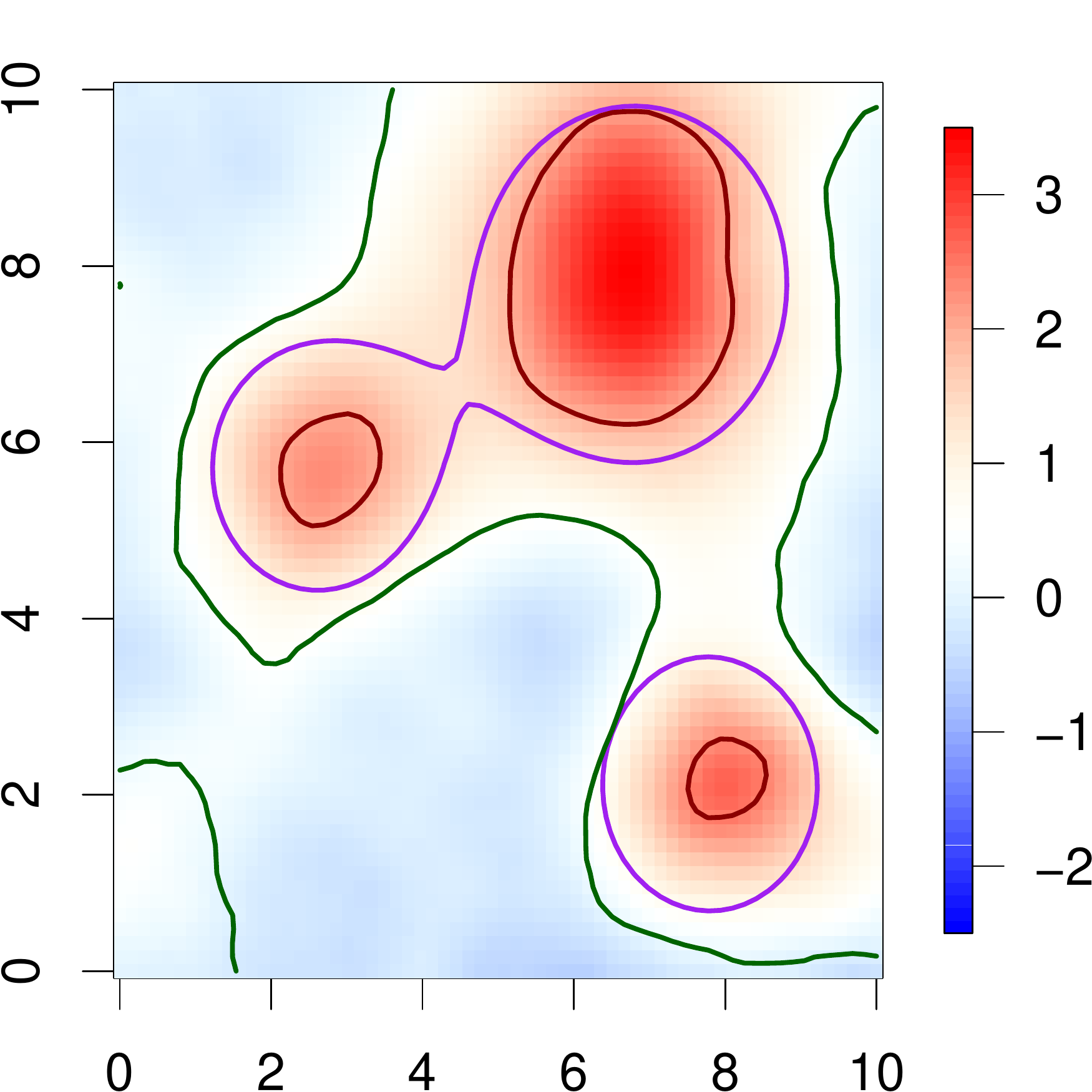}	& \includegraphics[width=0.25\textwidth]{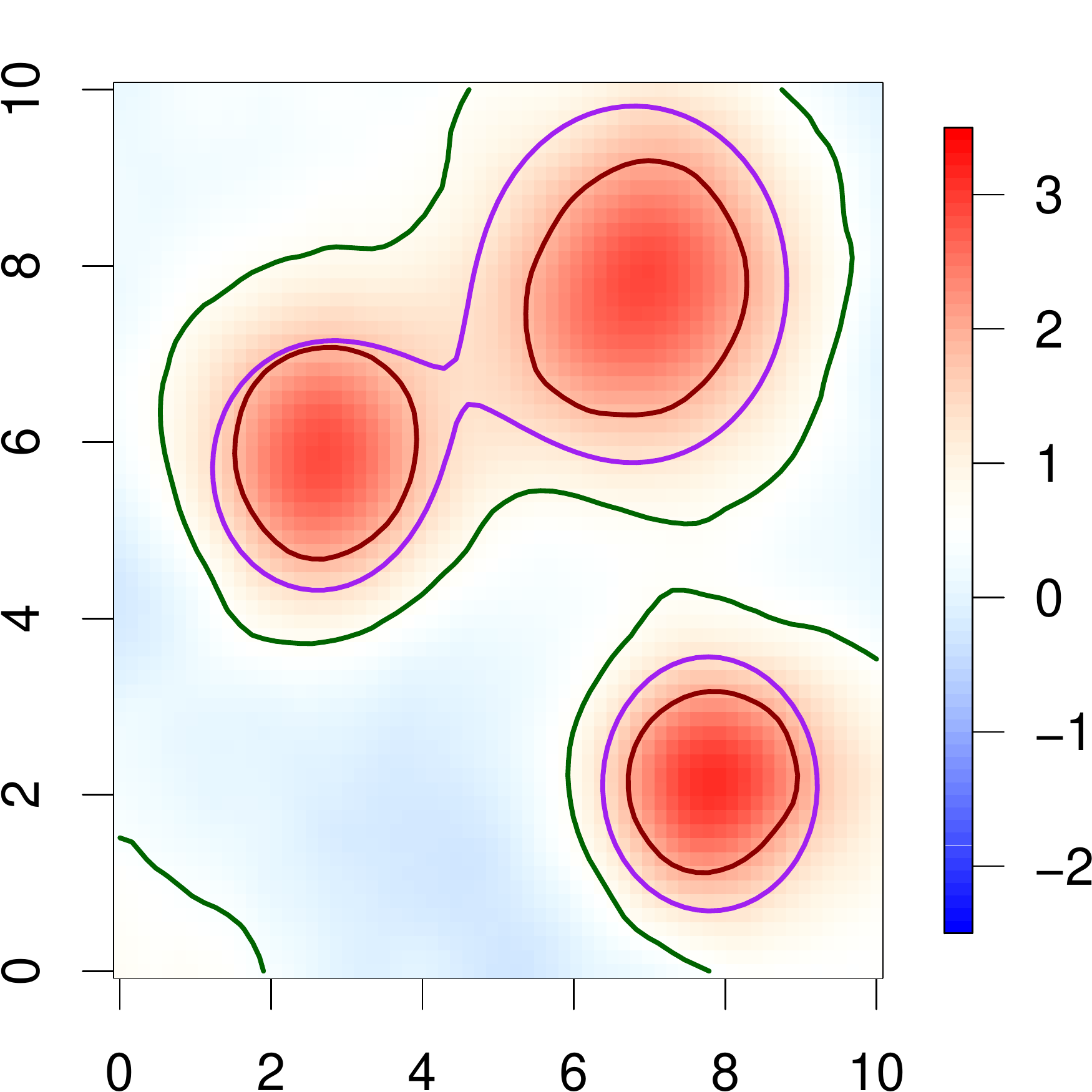} & \includegraphics[width=0.25\textwidth]{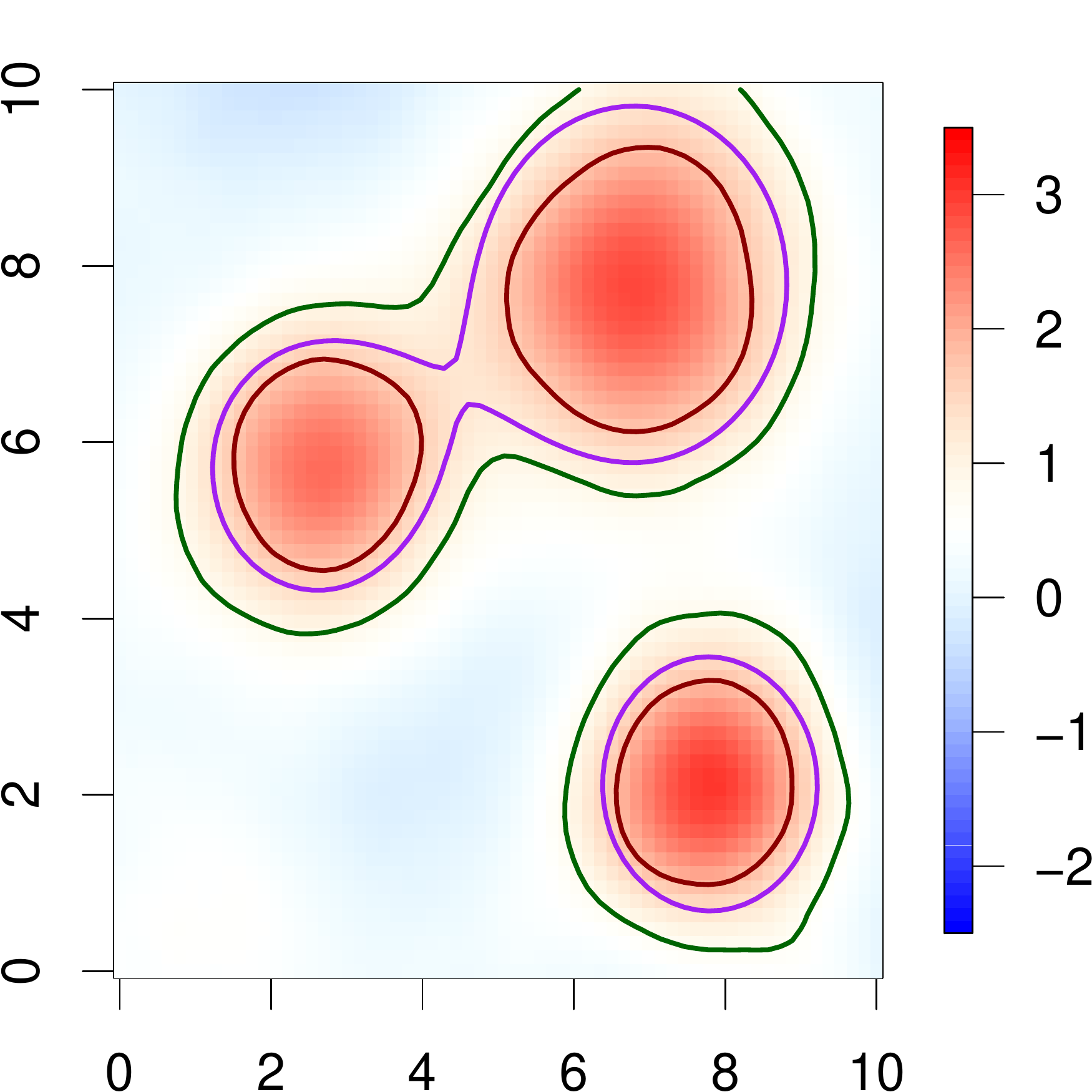} \\
\end{tabular} \centering{}\protect\caption{\label{fig:res_toy} The output of our method for the three noise
fields described above (corresponding to rows) and for sample sizes
$n=60,120,240$ (corresponding to columns) with the target function
$\mu(\protect\s)$ shown in Figure \ref{fig:mu_noise}. In all pictures
we show a heat map of the estimator $\protect\wh{\mu}_{n}(\protect\s)$,
the boundary of $A_{c}(\mu)$ in purple as well as the boundaries
of $\protect\wh A_{c}^{+}$ and $\protect\wh A_{c}^{-}$ in red and
green, respectively. The threshold $a$ was obtained according to
Theorem \ref{thm:inclusion} to guarantee inclusion $\protect\wh A_{c}^{+}\subset A_{c}\subset\protect\wh A_{c}^{-}$
with confidence $1-\alpha=0.9$.}
\end{figure}

\subsection{Performance of CoPE sets}

We analyzed the performance of our method on each $5000$ runs of
the toy examples shown in Section \ref{sub:Setup} with sample sizes
$n=60$, $n=120$ and $n=240$. Table \ref{table:coverage} shows
the percentage of trials in which coverage $\wh A_{c}^{+}\subset A_{c}\subset\wh A_{c}^{-}$
was achieved, if either the true boundary $\partial A_{c}$ or the
plug-in estimator $\partial\wh A_{c}$ was used to determine the threshold.

We see that the empirical coverage is smaller than the nominal level in all
experiments but approaches the nominal level  reasonably fast as the sample size
increases. In fact, when $n = 240$, the simulation confidence interval cover the
nominal level of 90\%, suggesting asymptotic unbiasedness. Comparing the two columns, we see that the non-asymptotic bias is not caused by the lack of knowledge of the true boundary. It may be a consequence of the bootstrap procedure instead.

\paragraph{Computational performance} As already noted in Section
\ref{sub:multiplier_bootstrap}, the multiplier bootstrap allows for a very fast
computation of CoPE sets. In the simulations, the CoPE sets for a sample of size
$n=240$, each on a grid of $64\times 64=4096$ locations could be computed in
less than two seconds on a standard laptop.

\begin{table}
\centering{}%
\begin{tabular}{|rrr|}
\hline 
 & $\partial A_{c}$ & $\partial\wh A_{c}$\tabularnewline
\hline 
Noise field 1 &  & \tabularnewline
$n=60$  & $86.20\%\pm0.49\%$  & $86.16\%\pm0.49\%$\tabularnewline
$120$  & $88.62\%\pm0.45\%$  & $88.74\%\pm0.45\%$\tabularnewline
$240$ & $88.94\%\pm0.44\%$ & $88.9\%\pm0.44\%$\tabularnewline
Noise field 2 &  & \tabularnewline
$n=60$  & $87.22\%\pm0.47\%$ & $88.74\%\pm0.45\%$\tabularnewline
$120$  & $89.22\%\pm0.44\%$ & $89.26\%\pm0.44\%$\tabularnewline
$240$ & $89.76\%\pm0.43\%$ & $89.70\%\pm0.43\%$\tabularnewline
Noise field 3 &  & \tabularnewline
$n=60$  & $86.44\%\pm0.48\%$ & $86.62\%\pm0.48\%$\tabularnewline
$120$  & $88.60\%\pm0.44\%$ & $88.78\%\pm0.45\%$\tabularnewline
$240$ & $89.76\%\pm0.43\%$ & $89.94\%\pm0.43\%$\tabularnewline
\hline 
\end{tabular}\protect\caption{\label{table:coverage} Percentage of trials in which $\protect\wh A_{c}^{+}\subset A_{c}\subset\protect\wh A_{c}^{-}$.
The nominal coverage probability is $90\%$.}
\end{table}

\subsection{Comparison with Taylor's Method}

In this Section we compare  the multiplier bootstrap with the method proposed by \cite{Taylor2007}, as described in Section \ref{sub:GKF}. We use both methods to approximate the distribution of $\sup_{\s\in\partial A_c}|\ve(\s)/\sigma(\s)|$, where $\ve(\s)$ is distributed according to \textsf{Noise 1} (see Section \ref{sub:Setup} above), $\sigma^2(\s)=\var[\ve(\s)]$ and $\partial A_c$ is the contour $A_c (\mu)$ of the function $\mu$ shown in Figure \ref{fig:mu_noise} at level $c=4/3$.  The true cumulative density function for $\sup_{\s\in\partial A_c}|\ve(\s)/\sigma(\s)|$ and  its empirical approximations based on  the multiplier bootstrap and Taylor's method are shown in Figure \ref{fig:MBvsT}. The empirical cdfs are each based on a single i.i.d. sample  $\ve_1(\s),\dots,\ve_n(\s)$ for $n=10,30$ and $60$. For the multiplier bootstrap we generated $5,000$ bootstrap realizations. The true cdf was calculated empirically using 10,000 i.i.d. samples of $\epsilon(s)$.

Both methods give a remarkably good approximation of the true distribution of the supremum, particularly for sample sizes of $n=30$ and higher. However, while  Taylor's method only gives a valid approximation in the tail of the distribution, the multiplier bootstrap approximates all parts of the cdf.

\begin{figure}
	\begin{centering}
	\includegraphics[width=\textwidth]{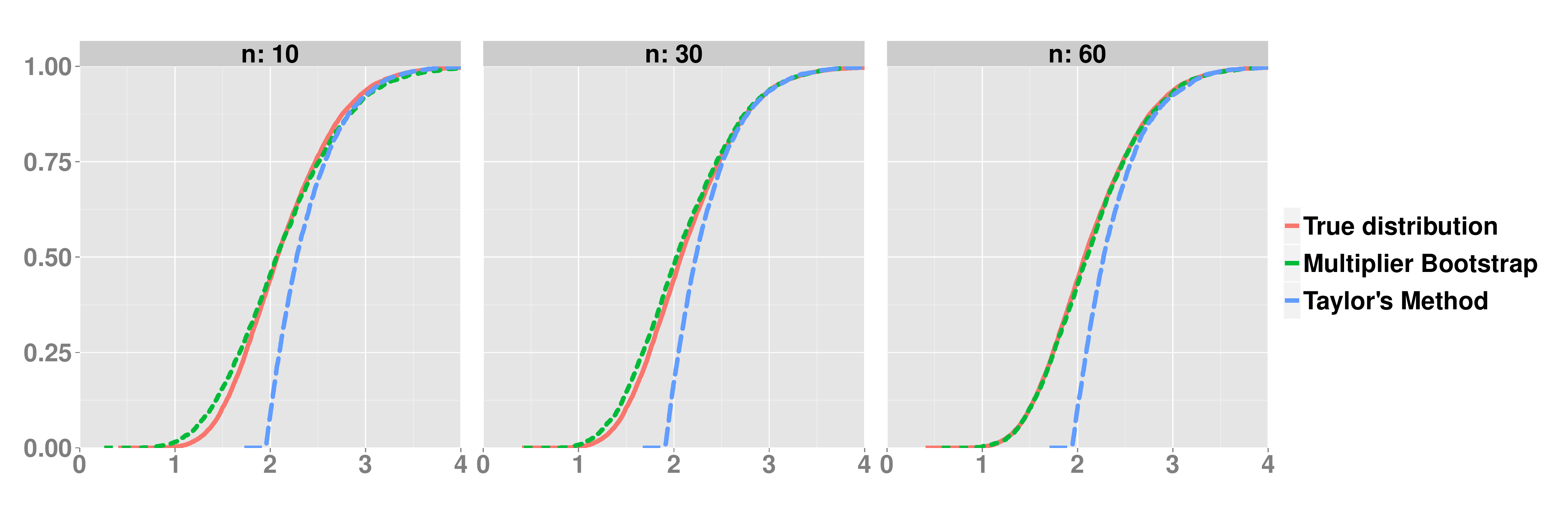}
	\protect\caption{\label{fig:MBvsT}The probability $P\left[\sup_{\s\in \partial A_c}|\ve(\s)/\sigma(\s)| \leq a \right]$ (the horizontal axis shows $a$) and approximations of it via the multiplier bootstrap or Taylor's method based on a sample of size $n$.  Here, the error field $\ve(\s)$ has the distribution described in \textsf{Noise 1} and $A_c = A_c(\mu)$ for the function $\mu$ shown in Figure \ref{fig:mu_noise}.}
	\end{centering}
\end{figure}

\section{\label{sec:Application-to-climate}Application to the climate data}

\subsection{Data setup}

In our application we have a total of $n=n^{(a)}+n^{(b)}$ observations,
the first $n^{(a)}$ observations are the 'past', the last $n^{(b)}$
are the 'future'. Within each period we model the change in mean temperature
linearly in time. More precisely we have 
\begin{align}
Y_{j} & (\s)=\Tmean^{(a)}(\s)+\Tslope^{(a)}(\s)t_{j}^{(a)}+\ve_{j}(\s),\quad j=1,\dots,n^{(a)}\nonumber \\
Y_{j} & (\s)=\Tmean^{(b)}(\s)+\Tslope^{(b)}(\s)t_{j}^{(b)}+\ve_{j}(\s),\quad j=n^{(a)}+1,\dots,n^{(a)}+n^{(b)}.\label{eq:temp_model_plain}
\end{align}
Without loss of generality, we may assume that $\sum_{j=1}^{n^{(a)}}t_{j}^{(a)}=0$
and $\sum_{j=n^{(a)}+1}^{n^{(a)}+n^{(b)}}t_{j}^{(b)}=0$. We will
denote the covariance of the error field $\ve(\s)$ by $\cv(\s_{1},\s_{2})=\cov\left[\ve(\s_{1}),\ve(\s_{2})\right]$.
Our goal is to give CoPE sets for the excursion sets of the
difference $\Tmean^{(b)}(\s)-\Tmean^{(a)}(\s)$. Therefore, we define
the parameter vector and the design matrix 
\[
\left(\begin{array}{c}
b_{1}(\s)\\
b_{2}(\s)\\
b_{3}(\s)\\
b_{4}(\s)
\end{array}\right)=\left(\begin{array}{c}
\Tmean^{(b)}(\s)-\Tmean^{(a)}(\s)\\
\Tmean^{(a)}(\s)\\
\Tslope^{(a)}(\s)\\
\Tslope^{(b)}(\s)
\end{array}\right),\quad\X=\left(\begin{array}{cccc}
0 & 1 & t_{1}^{(a)} & 0\\
\vdots & \vdots & \vdots & \vdots\\
0 & 1 & t_{n^{(a)}}^{(a)} & 0\\
1 & 1 & 0 & t_{n^{(a)}+1}^{(b)}\\
\vdots & \vdots & \vdots & \vdots\\
1 & 1 & 0 & t_{n^{(a)}+n^{(b)}}^{(b)}
\end{array}\right),
\]
to be able to rewrite \eqref{eq:temp_model_plain} as a general linear
model $\Y(\s)=\X\b(\s)+\e(\s).$ Our objective is now to formulate
Assumptions on the design and the noise under which we can apply Algorithm
\ref{Algorithm} to the data. This is done in the following.
\begin{assumption}
\label{ass:appl}Assume that 
\begin{enumerate}
	\item the parameter functions $\b$ are continuous and and the level set $\{\s:b_1(\s)=c\}$ is equal to $\partial A_c(b_1)$.
	\item the noise field $\ve(\s)$ has continuous sample paths with probability
	one and a centered unit variance Gaussian field with correlation
	function $\cv(\s_{1},\s_{2})$ also has continuous sample paths with probability
	one.
	\item the variance function $\sigma(\s)$ is continuous.
	\item there exist numbers $\delta,\beta>0$ and $\gamma\geq0$ such that
	the error field $\ve(\s)$ has the properties \textsf{\textup{N1}}-$\delta$
	and \textsf{\textup{N2}}-$(\gamma,\beta)$.
\item $n^{(a)}=n^{(b)}=\nicefrac{n}{2}$ and that both sets of design points
$t_{j}^{(a)}$ and $t_{j}^{(b)}$ are equally spaced (possibly with
different spacing for the periods $(a)$ and $(b)$). 
 
\end{enumerate}
\end{assumption}
The next final and central statement now asserts that these Assumptions
are indeed sufficient for Algorithm \ref{Algorithm} to be valid.
Its proof is a direct application of Theorem \ref{thm:weak-convergence}.
\begin{prop}
\label{prop:data_CLT}Under Model \eqref{eq:temp_model_plain} Assumptions
\ref{ass:appl} imply that Assumptions \ref{assu:CLT_and_mu_muhat}
hold for the target function $b_{1}(\s)$ and the estimator $\wh b_{1}(\s)$
with $\tau_{n}=2n^{-\nicefrac{1}{2}}$. In particular, 
\[
\frac{\sqrt{n}}{2\sigma(\s)}\left(\wh b_{1}(\s)-b_{1}(\s)\right)\ra G(\s),
\]
weakly, where $G$ is a mean zero unit variance Gaussian field with
correlation function $\cov\left[G(\s_{1}),G(\s_{2})\right]=\cv(\s_{1},\s_{2})$.
Consequently, Algorithm \ref{Algorithm} can be used to obtain CoPE
sets for the excursion sets of $b_{1}(\s)$.
\end{prop}

\subsection{Data analysis}

The results for the climate data described in the Introduction,  shown in Figure \ref{data_res}, correspond to CoPE sets for $b_{1}(\s)=\Tmean^{(b)}(\s)-\Tmean^{(a)}(\s)$
obtained via Algorithm \ref{Algorithm} with
$||\mathbf{v}||_{2}\pi_{n}^{\nicefrac{1}{2}}\sigma(\s)^{-1}=2/\sqrt{n}$
and $n=29+29=58$. The target level is $c=2^{\circ}C$ and nominal coverage probability is fixed at $1-\alpha=0.9$. 

For the mean summer temperature, it may be stated with $90\%$ confidence that the Rocky Mountains and the Sierra Madre Occidental mountains of Mexico are at risk of exhibiting a warming of $2^\circ C$ or more in the given time period, while  the Florida Peninsula, parts of the Mexican Gulf, large parts of the Canadian Northwest and the northern part of the Labrador Peninsula are not at risk.

For the mean winter temperature, some regions around the Hudson Bay and in the Canadian Shield are identified to be at a high risk while a comparatively small region north of the Mexican Gulf is considered not at risk for extreme warming.

For the computation time we remark that the entire analysis of one season,
including the pointwise linear regression and the multiplier bootstrap to obtain
the CoPE sets was performed in under five seconds on a regular laptop.

\subsection*{Acknowledgment}

M.S. acknowledges support by the ``Studienstiftung des Deutschen
Volkes'' and the SAMSI 2013-2014 program on Low-dimensional Structure
in High-dimensional Systems. A.S. and S.S. were partially supported by NIH grant
R01 CA157528. S.S. began working on this research  while he was a Scientist with the Institute for Mathematics Applied to the Geosciences, National Center for Atmospheric Research, Boulder, CO. All authors wish to thank
the North American Regional Climate Change Assessment Program (NARCCAP)
for providing the data used in this paper. NARCCAP is funded by the
National Science Foundation (NSF), the U.S. Department of Energy (DoE),
the National Oceanic and Atmospheric Administration (NOAA), and the
U.S. Environmental Protection Agency Office of Research and Development (EPA).

\nocite{*}
\bibliographystyle{plainnat}
\phantomsection\addcontentsline{toc}{section}{\refname}\bibliography{LvlSets}

\begin{thebibliography}{49}
\providecommand{\natexlab}[1]{#1}
\providecommand{\url}[1]{\texttt{#1}}
\expandafter\ifx\csname urlstyle\endcsname\relax
  \providecommand{\doi}[1]{doi: #1}\else
  \providecommand{\doi}{doi: \begingroup \urlstyle{rm}\Url}\fi

\bibitem[Adler(2000)]{Adler2000}
Robert~J. Adler.
\newblock {On Excursion Sets, Tube Formulas and Maxima of Random Fields}.
\newblock \emph{{The Annals of Applied Probability}}, {10}\penalty0
  (1):\penalty0 1--74, 2000.

\bibitem[Adler and Taylor(2007)]{Adler2007}
Robert~J. Adler and Jonathan~E Taylor.
\newblock \emph{{Random fields and geometry}}.
\newblock {Springer}, {New York}, 2007.

\bibitem[Adler et~al.(2012)Adler, Blanchet, and Liu]{Adler2012}
Robert~J. Adler, Jose~H. Blanchet, and Jingchen Liu.
\newblock {Efficient Monte Carlo for high excursions of Gaussian random
  fields}.
\newblock \emph{{The Annals of Applied Probability}}, {22}\penalty0
  (3):\penalty0 1167--1214, 2012.

\bibitem[Anderson and Bows(2011)]{Anderson2011}
Kevin Anderson and Alice Bows.
\newblock {Beyond `dangerous' climate change: emission scenarios for a new
  world}.
\newblock \emph{{Philosophical Transactions of the Royal Society A:
  Mathematical, Physical and Engineering Sciences}}, {369}\penalty0
  (1934):\penalty0 20--44, 2011.

\bibitem[Bassett and Koenker(1978)]{Bassett1978}
Jr. Bassett, Gilbert and Roger Koenker.
\newblock {Asymptotic Theory of Least Absolute Error Regression}.
\newblock \emph{{Journal of the American Statistical Association}},
  {73}\penalty0 (363):\penalty0 618--622, 1978.

\bibitem[Berman(1982)]{Berman1982}
Simeon~M. Berman.
\newblock {Sojourns and Extremes of Stationary Processes}.
\newblock \emph{{The Annals of Probability}}, {10}\penalty0 (1):\penalty0
  1--46, 1982.

\bibitem[Bickel and Wichura(1971)]{Bickel1971}
P.~J. Bickel and M.~J. Wichura.
\newblock {Convergence Criteria for Multiparameter Stochastic Processes and
  Some Applications}.
\newblock \emph{{The Annals of Mathematical Statistics}}, {42}\penalty0
  (5):\penalty0 1656--1670, 1971.

\bibitem[Bolin and Lindgren()]{Bolin2014}
David Bolin and Finn Lindgren.
\newblock {Excursion and contour uncertainty regions for latent Gaussian
  models}.
\newblock \emph{{Journal of the Royal Statistical Society: Series B
  (Statistical Methodology)}}.
\newblock To appear.

\bibitem[Cadre(2006)]{Cadre2006}
Beno\^{i}t Cadre.
\newblock {Kernel estimation of density level sets}.
\newblock \emph{{Journal of Multivariate Analysis}}, {97}\penalty0
  (4):\penalty0 999--1023, 2006.

\bibitem[Cavalier(1997)]{Cavalier1997}
Laurent Cavalier.
\newblock {Nonparametric Estimation of Regression Level Sets}.
\newblock \emph{{Statistics}}, {29}\penalty0 (2):\penalty0 131--160, 1997.

\bibitem[Chernozhukov et~al.(2013)Chernozhukov, Chetverikov, and
  Kato]{Chernozhukov2013}
Victor Chernozhukov, Denis Chetverikov, and Kengo Kato.
\newblock {Gaussian approximations and multiplier bootstrap for maxima of sums
  of high-dimensional random vectors}.
\newblock \emph{{The Annals of Statistics}}, {41}\penalty0 (6):\penalty0
  2786--2819, 2013.

\bibitem[Cuevas et~al.(2006)Cuevas, Gonz\'{a}lez-Manteiga, and
  Rodr\'{\i}guez-Casal]{Cuevas2006}
Antonio Cuevas, Wenceslao Gonz\'{a}lez-Manteiga, and Alberto
  Rodr\'{\i}guez-Casal.
\newblock {Plug-in Estimation of General Level Sets}.
\newblock \emph{{Australian \& New Zealand Journal of Statistics}},
  {48}\penalty0 (1):\penalty0 7--19, 2006.

\bibitem[Eicker(1963)]{Eicker1963}
F.~Eicker.
\newblock {Asymptotic Normality and Consistency of the Least Squares Estimators
  for Families of Linear Regressions}.
\newblock \emph{{The Annals of Mathematical Statistics}}, {34}\penalty0
  (2):\penalty0 447--456, 1963.

\bibitem[Flato(2005)]{Flato2005}
G.~M. Flato.
\newblock {The third generation coupled global climate model (CGCM3)}.
\newblock \emph{{Available on line at
  \url{http://www.cccma.bc.ec.gc.ca/models/cgcm3.shtml}}}, 2005.

\bibitem[French(2014)]{French2014}
Joshua~P. French.
\newblock {Confidence regions for the level curves of spatial data}.
\newblock \emph{{Environmetrics}}, {25}\penalty0 (7):\penalty0 498--512, 2014.

\bibitem[French and Sain(2013)]{French2013}
Joshua~P. French and Stephan~R. Sain.
\newblock {Spatio-temporal exceedance locations and confidence regions}.
\newblock \emph{{The Annals of Applied Statistics}}, {7}\penalty0 (3):\penalty0
  1421--1449, 2013.

\bibitem[Genovese et~al.(2002)Genovese, Lazar, and Nichols]{Genovese2002}
Christopher~R. Genovese, Nicole~A. Lazar, and Thomas Nichols.
\newblock {Thresholding of Statistical Maps in Functional Neuroimaging Using
  the False Discovery Rate}.
\newblock \emph{{NeuroImage}}, {15}\penalty0 (4):\penalty0 870--878, 2002.

\bibitem[Hardle and Mammen(1993)]{Hardle1993}
W.~Hardle and E.~Mammen.
\newblock {Comparing Nonparametric Versus Parametric Regression Fits}.
\newblock \emph{{The Annals of Statistics}}, {21}\penalty0 (4):\penalty0
  1926--1947, 1993.

\bibitem[Khoshnevisan(2002)]{Khoshnevisan2002}
Davar Khoshnevisan.
\newblock \emph{{Multiparameter Processes: an introduction to random fields}}.
\newblock {Springer}, 2002.

\bibitem[Lindgren and Rychlik(1995)]{Lindgren1995}
Georg Lindgren and Igor Rychlik.
\newblock {How reliable are contour curves? Confidence sets for level
  contours}.
\newblock \emph{{Bernoulli}}, {1}\penalty0 (4):\penalty0 301--319, 1995.

\bibitem[Mammen(1992)]{Mammen1992}
Enno Mammen.
\newblock {Bootstrap, wild bootstrap, and asymptotic normality}.
\newblock \emph{{Probability Theory and Related Fields}}, {93}\penalty0
  (4):\penalty0 439--455, 1992.

\bibitem[Mammen(1993)]{Mammen1993}
Enno Mammen.
\newblock {Bootstrap and Wild Bootstrap for High Dimensional Linear Models}.
\newblock \emph{{The Annals of Statistics}}, {21}\penalty0 (1):\penalty0
  255--285, 1993.

\bibitem[Mammen and Polonik(2013)]{Mammen2013}
Enno Mammen and Wolfgang Polonik.
\newblock {Confidence regions for level sets}.
\newblock \emph{{Journal of Multivariate Analysis}}, {122}:\penalty0 202--214,
  2013.

\bibitem[Mason and Polonik(2009)]{Mason2009}
David~M. Mason and Wolfgang Polonik.
\newblock {Asymptotic normality of plug-in level set estimates}.
\newblock \emph{{The Annals of Applied Probability}}, {19}\penalty0
  (3):\penalty0 1108--1142, 2009.

\bibitem[Mearns et~al.(2013)Mearns, Sain, Leung, Bukovsky, McGinnis, Biner,
  Caya, Arritt, Gutowski, Takle, Snyder, Jones, Nunes, Tucker, Herzmann,
  McDaniel, and Sloan]{Mearns2013}
L.~O. Mearns, S.~Sain, L.~R. Leung, M.~S. Bukovsky, S.~McGinnis, S.~Biner,
  D.~Caya, R.~W. Arritt, W.~Gutowski, E.~Takle, M.~Snyder, R.~G. Jones,
  A.~M.~B. Nunes, S.~Tucker, D.~Herzmann, L.~McDaniel, and L.~Sloan.
\newblock {Climate change projections of the North American Regional Climate
  Change Assessment Program (NARCCAP)}.
\newblock \emph{{Climatic Change}}, {120}\penalty0 (4):\penalty0 965--975,
  2013.

\bibitem[Mearns et~al.(2009)Mearns, Gutowski, Jones, Leung, McGinnis, Nunes,
  and Qian]{Mearns2009}
Linda~O. Mearns, William Gutowski, Richard Jones, Ruby Leung, Seth McGinnis,
  Ana Nunes, and Yun Qian.
\newblock {A Regional Climate Change Assessment Program for North America}.
\newblock \emph{{Eos, Transactions American Geophysical Union}}, {90}\penalty0
  (36):\penalty0 311--311, 2009.

\bibitem[Mearns et~al.(2012)Mearns, Arritt, Biner, Bukovsky, McGinnis, Sain,
  Caya, Correia, Flory, Gutowski, Takle, Jones, Leung, Moufouma-Okia, McDaniel,
  Nunes, Qian, Roads, Sloan, and Snyder]{Mearns2012}
Linda~O. Mearns, Ray Arritt, S\'{e}bastien Biner, Melissa~S. Bukovsky, Seth
  McGinnis, Stephan Sain, Daniel Caya, James Correia, Dave Flory, William
  Gutowski, Eugene~S. Takle, Richard Jones, Ruby Leung, Wilfran Moufouma-Okia,
  Larry McDaniel, Ana M.~B. Nunes, Yun Qian, John Roads, Lisa Sloan, and Mark
  Snyder.
\newblock {The North American Regional Climate Change Assessment Program:
  Overview of Phase I Results}.
\newblock \emph{{Bulletin of the American Meteorological Society}},
  {93}\penalty0 (9):\penalty0 1337--1362, 2012.

\bibitem[Michalakes et~al.(2004)Michalakes, Dudhia, Gill, Henderson, Klemp,
  Skamarock, and Wang]{Michalakes2004}
J.~Michalakes, J.~Dudhia, D.~Gill, T.~Henderson, J.~Klemp, W.~Skamarock, and
  W.~Wang.
\newblock {The weather research and forecast model: software architecture and
  performance}.
\newblock In \emph{{Proceedings of the 11th ECMWF Workshop on the Use of High
  Performance Computing In Meteorology}}, volume~{25}. {World Scientific},
  2004.

\bibitem[Pham(2013)]{Pham2013}
Viet-Hung Pham.
\newblock {On the rate of convergence for central limit theorems of sojourn
  times of Gaussian fields}.
\newblock \emph{{Stochastic Processes and their Applications}}, {123}\penalty0
  (6):\penalty0 2158--2174, 2013.

\bibitem[Polfeldt(1999)]{Polfeldt1999}
Thomas Polfeldt.
\newblock {On the quality of contour maps}.
\newblock \emph{{Environmetrics}}, {10}\penalty0 (6):\penalty0 785--790, 1999.

\bibitem[{R Core Team}(2014)]{R2014}
{R Core Team}.
\newblock \emph{R: A Language and Environment for Statistical Computing}.
\newblock R Foundation for Statistical Computing, Vienna, Austria, 2014.
\newblock URL \url{http://www.R-project.org/}.

\bibitem[Rao and Toutenburg(1995)]{Rao1995}
Calyampudi~Radhakrishna Rao and Helge Toutenburg.
\newblock \emph{{Linear models}}.
\newblock {Springer}, 1995.

\bibitem[Rigollet and Vert(2009)]{Rigollet2009}
Philippe Rigollet and R\'{e}gis Vert.
\newblock {Optimal rates for plug-in estimators of density level sets}.
\newblock \emph{{Bernoulli}}, {15}\penalty0 (4):\penalty0 1154--1178, 2009.

\bibitem[Rogelj et~al.(2009)Rogelj, Hare, Nabel, Macey, Schaeffer, Markmann,
  and Meinshausen]{Rogelj2009}
Joeri Rogelj, Bill Hare, Julia Nabel, Kirsten Macey, Michiel Schaeffer,
  Kathleen Markmann, and Malte Meinshausen.
\newblock {Halfway to Copenhagen, no way to 2 {\textdegree}C}.
\newblock \emph{{Nature Reports Climate Change}}, \penalty0 (0907):\penalty0
  81--83, 2009.

\bibitem[Schwartzman and Lin(2011)]{Schwartzman2011}
Armin Schwartzman and Xihong Lin.
\newblock {The effect of correlation in false discovery rate estimation}.
\newblock \emph{{Biometrika}}, {98}\penalty0 (1):\penalty0 199--214, 2011.

\bibitem[Schwartzman et~al.(2010)Schwartzman, Dougherty, and
  Taylor]{Schwartzman2010}
Armin Schwartzman, Robert~F. Dougherty, and Jonathan~E. Taylor.
\newblock {Group Comparison of Eigenvalues and Eigenvectors of Diffusion
  Tensors}.
\newblock \emph{{Journal of the American Statistical Association}},
  {105}\penalty0 (490):\penalty0 588--599, 2010.

\bibitem[Shiohama and Xu(2011)]{Shiohama2011}
Katsuhiro Shiohama and Hong-Wei Xu.
\newblock {An Integral Formula for Lipschitz-Killing Curvature and the Critical
  Points of Height Functions}.
\newblock \emph{{Journal of Geometric Analysis}}, {21}\penalty0 (2):\penalty0
  241--251, 2011.

\bibitem[Singh et~al.(2009)Singh, Scott, and Nowak]{Singh2009}
Aarti Singh, Clayton Scott, and Robert Nowak.
\newblock {Adaptive Hausdorff estimation of density level sets}.
\newblock \emph{{The Annals of Statistics}}, {37}\penalty0 (5B):\penalty0
  2760--2782, 2009.

\bibitem[Sommerfeld(2015)]{cope2015}
Max Sommerfeld.
\newblock \emph{cope: Coverage Probability Excursion (CoPE) sets.}, 2015.
\newblock URL \url{http://www.cran.r-project.org/package=cope}.
\newblock R package version 0.1.

\bibitem[Taylor et~al.(2005)Taylor, Takemura, and Adler]{Taylor2005}
Jonathan Taylor, Akimichi Takemura, and Robert~J. Adler.
\newblock {Validity of the expected Euler characteristic heuristic}.
\newblock \emph{{The Annals of Probability}}, {33}\penalty0 (4):\penalty0
  1362--1396, 2005.

\bibitem[Taylor(2006)]{Taylor2006}
Jonathan~E. Taylor.
\newblock {A Gaussian kinematic formula}.
\newblock \emph{{The Annals of Probability}}, {34}\penalty0 (1):\penalty0
  122--158, 2006.

\bibitem[Taylor and Adler(2003)]{Taylor2003}
Jonathan~E. Taylor and Robert~J. Adler.
\newblock {Euler Characteristics for Gaussian Fields on Manifolds}.
\newblock \emph{{The Annals of Probability}}, {31}\penalty0 (2):\penalty0
  533--563, 2003.

\bibitem[Taylor and Adler(2009)]{Taylor2009}
Jonathan~E. Taylor and Robert~J. Adler.
\newblock {Gaussian processes, kinematic formulae and Poincar\'e's limit}.
\newblock \emph{{The Annals of Probability}}, {37}\penalty0 (4):\penalty0
  1459--1482, 2009.

\bibitem[Taylor and Worsley(2007)]{Taylor2007}
Jonathan~E. Taylor and Keith~J. Worsley.
\newblock {Detecting Sparse Signals in Random Fields, with an Application to
  Brain Mapping}.
\newblock \emph{{Journal of the American Statistical Association}},
  {102}\penalty0 (479):\penalty0 913--928, 2007.

\bibitem[Tsybakov(1997)]{Tsybakov1997}
A.~B. Tsybakov.
\newblock {On nonparametric estimation of density level sets}.
\newblock \emph{{The Annals of Statistics}}, {25}\penalty0 (3):\penalty0
  948--969, 1997.

\bibitem[Wameling(2003)]{Wameling2003}
Almuth Wameling.
\newblock {Accuracy of geostatistical prediction of yearly precipitation in
  Lower Saxony}.
\newblock \emph{{Environmetrics}}, {14}\penalty0 (7):\penalty0 699--709, 2003.

\bibitem[Willett and Nowak(2007)]{Willett2007}
R.M. Willett and R.D. Nowak.
\newblock {Minimax Optimal Level-Set Estimation}.
\newblock \emph{{IEEE Transactions on Image Processing}}, {16}\penalty0
  (12):\penalty0 2965--2979, 2007.

\bibitem[Worsley et~al.(1996)Worsley, Marrett, Neelin, Vandal, Friston, Evans,
  et~al.]{Worsley1996}
Keith~J. Worsley, Sean Marrett, Peter Neelin, Alain~C. Vandal, Karl~J. Friston,
  Alan~C. Evans, et~al.
\newblock {A unified statistical approach for determining significant signals
  in images of cerebral activation}.
\newblock \emph{{Human brain mapping}}, {4}\penalty0 (1):\penalty0 58--73,
  1996.

\bibitem[Wu(1986)]{Wu1986}
C.~F.~J. Wu.
\newblock {Jackknife, Bootstrap and Other Resampling Methods in Regression
  Analysis}.
\newblock \emph{{The Annals of Statistics}}, {14}\penalty0 (4):\penalty0
  1261--1295, 1986.

\end{thebibliography}

\appendix
\section{Proofs}
\begin{proof}[Proof of Theorem \ref{thm:inclusion}]
We start by showing that 
\begin{equation}
\liminf_{n\ra\infty}P\left[\wh A_{c}^{+}\subset A_{c}\subset\wh A_{c}^{-}\right]\geq P\left[\sup_{\partial A_{c}}|G(s)|\leq a\right].\label{eq:liminf}
\end{equation}
For $\eta>0$ define the inflated boundary $A_{c}^{\eta}=\{\s\in S:c-\eta\sigma(\s)\leq\mu(\s)\leq c+\eta\sigma(\s)\}$. The idea of the proof is that, loosely speaking,  points outside of $A_{c}^{\eta}$ become irrelevant in the limit $n\ra\infty$ since their values are far from $c$ and, if we let $\eta$ go to zero at an appropriate rate, we finally end up with the boundary $\dAc$.
More precisely, we note that 
\begin{align*}
 & \frac{\wh{\mu}_{n}(\s)-\mu(\s)}{\tau_{n}\sigma(\s)}\geq-a\text{ for all }\s\in A_{c}\cap A_{c}^{\eta}\text{ and }\frac{\wh{\mu}_{n}(\s)-\mu(\s)}{\tau_{n}\sigma(\s)}\geq-\eta\tau_{n}^{-1}-a\text{ for all }\s\in A_{c}\setminus A_{c}^{\eta},
\end{align*}
implies that $\wh{\mu}_{n}(\s)\geq c-\tau_{n}\sigma(\s)a$ for all $\s\in A_{c}$ and hence $A_{c}\subset\wh A_{c}^{-}$.
Similarly, 
\begin{align*}
 \frac{\wh{\mu}_{n}(\s)-\mu(\s)}{\tau_{n}\sigma(\s)}<a \text{ for all }\s\in\left(S\setminus A_{c}\right)\cap A_{c}^{\eta} \text{ and }\frac{\wh{\mu}_{n}(\s)-\mu(\s)}{\tau_{n}\sigma(\s)}<\eta\tau_{n}^{-1}+a \text{ for all }\s\in\left(S\setminus A_{c}\right)\setminus A_{c}^{\eta}
\end{align*}
implies $A_{c}^{+}\subset A_{c}$.
Combining these observations, we see that $\wh A_{c}^{+}\subset A_{c}\subset\wh A_{c}^{-}$ holds, provided that $\sup_{\s\in A_{c}^{\eta_{n}}}\left|\frac{\wh{\mu}_{n}(\s)-\mu(\s)}{\tau_{n}\sigma(\s)}\right|<a$ and $\sup_{\s\in S\setminus A_{c}^{\eta_{n}}}\left|\frac{\wh{\mu}_{n}(\s)-\mu(\s)}{\tau_{n}\sigma(\s)}\right|<a+\eta_{n}\tau_{n}^{-1}$.
Now, let $\left\{ \eta_{n}\right\} _{n\in\NN}$ be a sequence of positive
numbers such that $\eta_{n}\ra0$ and $\eta_{n}\tau_{n}^{-1}\ra\infty$.
We can then write 
\begin{align}
 & P\left[\wh A_{c}^{+}\subset A_{c}\subset\wh A_{c}^{-}\right]\geq P\left[\sup_{\s\in A_{c}^{\eta_{n}}}\left|\frac{\wh{\mu}_{n}(\s)-\mu(\s)}{\tau_{n}\sigma(\s)}\right|<a\mbox{ and }\sup_{\s\in S\setminus A_{c}^{\eta_{n}}}\left|\frac{\wh{\mu}_{n}(\s)-\mu(\s)}{\tau_{n}\sigma(\s)}\right|<a+\eta_{n}\tau_{n}^{-1}\right]\nonumber \\
 & \geq\underset{(I)}{\underbrace{P\left[\sup_{\s\in A_{c}^{\eta_{n}}}\left|\frac{\wh{\mu}_{n}(\s)-\mu(\s)}{\tau_{n}\sigma(\s)}\right|<a\right]}}+\underset{(II)}{\underbrace{P\left[\sup_{\s\in S\setminus A_{c}^{\eta_{n}}}\left|\frac{\wh{\mu}_{n}(\s)-\mu(\s)}{\tau_{n}\sigma(\s)}\right|<a+\eta_{n}\tau_{n}^{-1}\right]-1}}.\label{eq:proof_sup_split}
\end{align}
We first show that the term (II) goes to zero. To this end let $\delta>0$
arbitrary. Let $b\in\RR$ such that $P\left[\sup_{\s\in S}|G(\s)|<b\right]\geq1-\delta$
and $n_{0}\in\NN$ such that $a+\eta_{n}\tau_{n}^{-1}\geq b$ for
all $n\geq n_{0}$. Also, let $n_{1}$ large enough such that 
\[
\left|P\left[\sup_{\s\in S}\left|\frac{\wh{\mu}_{n}(\s)-\mu(\s)}{\tau_{n}\sigma(\s)}\right|<b\right]-P\left[\sup_{\s\in S}|G(\s)|<b\right]\right|<\delta,
\]
for all $n\geq n_{1}$. In consequence, for all $n\geq\max\left\{ n_{0},n_{1}\right\} $
\begin{align*}
P\left[\sup_{\s\in S\setminus A_{c}^{\eta_{n}}}\left|\frac{\wh{\mu}_{n}(\s)-\mu(\s)}{\tau_{n}\sigma(\s)}\right|<a+\eta_{n}\tau_{n}^{-1}\right] & \geq P\left[\sup_{\s\in S}\left|\frac{\wh{\mu}_{n}(\s)-\mu(\s)}{\tau_{n}\sigma(\s)}\right|<b\right] \geq1-2\delta.
\end{align*}
Since $\delta>0$ was arbitrary it follows that (II) converges to
zero as $n\ra\infty$.

To prove convergence of (I) we need the following 
\begin{lem}
\label{lem:Hausdorff}Under Assumptions \ref{assu:CLT_and_mu_muhat}
part (a) if $\eta_{n}\ra0$ then the Hausdorff distance
$\delta_{n}:=d_{H}\left(A_{c}^{\eta_{n}},\partial A_{c}\right)\ra0$.
\end{lem}
\begin{proof}
	Let us define the set $(\partial A_c)_\varepsilon:=\{\s\in S:d(\s,\dAc)<\varepsilon\}$. We prove the assertion by showing that for any $\varepsilon>0$ there exists an $\eta>0$ such that $A_c^\eta\subset (\dAc)_\varepsilon$. To this end, assume the contrary. Then, there exists  $\varepsilon>0$ such that for any $\eta>0$ we find $\s_\eta\in A_c^\eta$ with $d(\s_\eta, \dAc)\geq\varepsilon$. The sequence $(s_\eta)_{\eta\downarrow 0}$ is contained in the compact set $S$ and hence has a convergent subsequence with limit $\s^*$, say. By construction, we have $\s^*\in\cap_{\eta>0}A_c^\eta =\dAc$. On the other hand, $0=d(\s^*,\dAc)=\lim_{\eta\ra0}d(\s_\eta,\dAc)\geq \varepsilon$, a contradiction.
\end{proof}
Recall that for a function $f:S\ra\RR$ and some number $\delta>0$
the modulus of continuity is defined as 
$w(f,\delta)=\sup_{|\s_{1}-\s_{2}|\leq\delta}\left|f(\s_{1})-f(\s_{2})\right|$.
Since $\left\{ \tau_{n}^{-1}\sigma(\s)^{-1}\left(\wh{\mu}_{n}(\s)-\mu(\s)\right)\right\} _{n\in\NN}$
is weakly convergent, we have 
\begin{equation}
\lim_{\delta\ra0}\limsup_{n\ra\infty}P\left[w\left(\frac{\wh{\mu}_{n}(\s)-\mu(\s)}{\tau_{n}\sigma(\s)},\delta\right)\geq\zeta\right]=0
\end{equation}
for all positive $\zeta$  \cite[Prop. 2.4.1 and Exc. 3.3.1]{Khoshnevisan2002}. Together with Lemma \ref{lem:Hausdorff} this implies
\[
\left|\sup_{\s\in A_{c}^{\eta_{n}}}\left|\frac{\wh{\mu}_{n}(\s)-\mu(\s)}{\tau_{n}\sigma(\s)}\right|-\sup_{\s\in\partial A_{c}}\left|\frac{\wh{\mu}_{n}(\s)-\mu(\s)}{\tau_{n}\sigma(\s)}\right|\right|\leq w\left(\frac{\wh{\mu}_{n}(\s)-\mu(\s)}{\tau_{n}\sigma(\s)},\delta_{n}\right)
\ra 0,
\]
in probability.
Since $\sup_{\s\in\partial A_{c}}\left|(\wh{\mu}_{n}(\s)-\mu(\s)) / (\tau_{n}\sigma(\s))\right|$
converges in distribution to $\sup_{\s\in\partial A_{c}}\left|G(\s)\right|$
this yields 
\[
\sup_{\s\in A_{c}^{\eta_{n}}}\left|\frac{\wh{\mu}_{n}(\s)-\mu(\s)}{\tau_{n}\sigma(\s)}\right|\ra\sup_{\s\in\partial A_{c}}\left|G(\s)\right|
\]
in distribution. In view of \eqref{eq:proof_sup_split} this completes
the proof of \eqref{eq:liminf}.

It remains to prove the opposite inequality, i.e. 
\begin{equation}
\limsup_{n\ra\infty}P\left[\wh A_{c}^{+}\subset A_{c}\subset\wh A_{c}^{-}\right]\leq P\left[\sup_{\partial A_{c}}|G(s)|\leq a\right].\label{eq:limsup}
\end{equation}
If for some arbitrary $\delta>0$ we have $\tau_{n}^{-1}\sigma(\s)^{-1}\left(\wh{\mu}_{n}(\s)-c\right)\geq a+\delta$
for some $\s\in\partial A_{c}$ then by continuity there is a $\s\in S\setminus A_{c}$
for which $\tau_{n}^{-1}\sigma(\s)^{-1}\left(\wh{\mu}_{n}(\s)-c\right)\geq a$
and hence the inclusion $\wh A_{c}^{+}\subset A_{c}$ does not hold.
Since an analogous argument works for the inclusion $A_{c}\subset\wh A_{c}^{-}$,
we have 
\begin{align*}
  P\left[\wh A_{c}^{+}\subset A_{c}\subset\wh A_{c}^{-}\right]
 & \leq1-P\left[\exists\s\in\partial A_{c}:\frac{\wh{\mu}_{n}(\s)-c}{\tau_{n}\sigma(\s)}\geq a+\delta\mbox{ or }\frac{\wh{\mu}_{n}(\s)-c}{\tau_{n}\sigma(\s)}\leq-a-\delta\right]\\
 & \leq1-P\left[\sup_{\s\in\partial A_{c}}\left|\frac{\wh{\mu}_{n}(\s)-c}{\tau_{n}\sigma(\s)}\right|\geq a+\delta\right]
  \ra P\left[\sup_{\s\in\partial A_{c}}\left|G(\s)\right|\leq a+\delta\right].
\end{align*}
Since $\delta>0$ was arbitrary and $\sup_{\s\in\partial A_{c}}|G(\s)|$
has a continuous distribution the bound \eqref{eq:limsup} follows.
\end{proof}

\begin{proof}[Proof of Corollary \ref{cor:inclusion_contour}]
	For any pair of nested sets $\wh{A}_c^\pm$ we have that $\wh A_{c}^{+}\subset A_{c}\subset\wh A_{c}^{-}$ implies $\partial A_c\subset \cl(\wh{A}_c^-\setminus\wh{A}_c^+)$. On the other hand, the latter will certainly fail to hold if $\sup_{\s\in\dAc}\left|\frac{\wh{\mu}_n(s)-c}{\tau_n\sigma(\s)}\right|>a$. Combining these two observations yields
	\[
	P\left[\wh A_{c}^{+}\subset A_{c}\subset\wh A_{c}^{-}\right]\leq P\left[\partial A_c\subset \cl(\wh{A}_c^-\setminus\wh{A}_c^+)\right]
	\leq P\left[\left|\sup_{\s\in\dAc}\frac{\wh{\mu}_n(s)-c}{\tau_n\sigma(\s)}\right|\leq a\right].
	\]
	Taking the limit $n\ra\infty$ of this inequality and using Theorem \ref{thm:inclusion} gives the assertion.
\end{proof}

\begin{proof}[Proof of Theorem \ref{thm:weak-convergence}]
We begin by proving part (a). Let us define $\mathbf{A}=\left(\X^{T}\X\right)^{-1}\X^{T}$, giving
$\wh{\b}(\s)-\b(\s)=\mathbf{A}\e(\s)$. In order to prove weak convergence
of the process we first show convergence of the finite dimensional
distributions and then tightness of the sequence \cite[Prop. 3.3.1]{Khoshnevisan2002}.

For the former, let $\s_{1},\dots,\s_{K}\in S$ be arbitrary. We need
to show that with 
\begin{align*}
U & =\left(\sqrt{\X^{T}\X}\sigma(\s_{1})^{-1}\left(\wh{\b}(\s_{1})-\b(\s_{1})\right),\dots,\sqrt{\X^{T}\X}\sigma(\s_{K})^{-1}\left(\wh{\b}(\s_{K})-\b(\s_{K})\right)\right)\\
 & =\mathbf{I}_{K}\otimes\left[\sqrt{\X^{T}\X}\right]\left(\sigma(\s_{1})^{-1}\mathbf{A}\e(\s_{1}),\dots,\sigma(\s_{K})^{-1}\mathbf{A}\e(\s_{K})\right),
\end{align*}
(here, $'\otimes'$ denotes the Kronecker product of two matrices)
we have convergence $U\ra\mathcal{N}(0,\left[\cv(\s_{i},\s_{j})\right]_{i,j=1}^{K}\otimes\mathbf{I}_{p})$
in distribution. We readily see that $E\left[U\right]=0$ and for
the covariance we compute 
\begin{align*}
  \cov\left[U\right]
 & =\mathbf{I}_{K}\otimes\left[\sqrt{\X^{T}\X}\mathbf{A}\right]E\left\{ \left(\sigma(\s_{1})^{-1}\e(\s_{1}),\dots,\sigma(\s_{K})^{-1}\e(\s_{K})\right)\left(\sigma(\s_{1})^{-1}\e(\s_{1}),\dots,\sigma(\s_{K})^{-1}\e(\s_{K})\right)^{T}\right\} \\
 & \quad\times\mathbf{I}_{K}\otimes\left[\mathbf{A}^{T}\sqrt{\X^{T}\X}\right]\\
 & =\left\{ \mathbf{I}_{K}\otimes\left[\sqrt{\X^{T}\X}\mathbf{A}\right]\right\} \left\{ \left[\cv(\s_{i},\s_{j})\right]_{i,j=1}^{K}\otimes\mathbf{I}_{n}\right\} \left\{ \mathbf{I}_{K}\otimes\left[\mathbf{A}^{T}\sqrt{\X^{T}\X}\right]\right\} \\
 & =\left[\cv(\s_{i},\s_{j})\right]_{i,j=1}^{K}\otimes\sqrt{\X^{T}\X}\mathbf{A}\mathbf{A}^{T}\sqrt{\X^{T}\X}
  =\left[\cv(\s_{i},\s_{j})\right]_{i,j=1}^{K}\otimes\mathbf{I}_{p}.
\end{align*}
We employ the Cramér-Wold device to show convergence of $U$. Indeed,
let and $(\alpha_{1},\dots,\alpha_{K})\in\RR^{K\times p}$ be some
fixed arbitrary vector and compute
\begin{align*}
\left\langle U,\alpha\right\rangle  & =\big\langle\left(\sigma(\s_{1})^{-1}\e(\s_{1}),\dots,\sigma(\s_{K})^{-1}\e(\s_{K})\right),\mathbf{I}_{K}\otimes\left[\mathbf{A}^{T}\sqrt{\X^{T}\X}\right]\alpha\big\rangle\\
 & =\sum_{i=1}^{K}\sigma(\s_{i})^{-1}\langle\e(\s_{i}),\mathbf{A}^{T}\sqrt{\X^{T}\X}\alpha_{i}\rangle
  =\sum_{i=1}^{K}\sum_{j=1}^{n}\sigma(\s_{i})^{-1}\ve_{j}(\s_{i})\left(\mathbf{A}^{T}\sqrt{\X^{T}\X}\alpha_{i}\right)_{j}.
\end{align*}
By interchanging the sums and defining 
$W_{j}=\sum_{i=1}^{K}\sigma(\s_{i})^{-1}\ve_{j}(\s_{i})\left(A^{T}\sqrt{\X^{T}\X}\alpha_{i}\right)_{j}$,
we have managed to write $\langle U,\alpha\rangle$ as a sum of independent
random variables $\langle U,\alpha\rangle=\sum_{j}W_{j}$. The goal
is now to use the CLT in the form of Lyapunov for the random variables
$W_{j}$. To this end compute 
$\var\left[\sum_{j=1}^{n}W_{j}\right]  =\var\left[\langle U,\alpha\rangle\right]=\alpha^{T}\cov\left[U\right]\alpha$ 
and note that since we have already showed $U$ to have the right
covariance the claimed convergence will follow once we establish the
Lyapunov condition. For this purpose let $\delta$ be as in Assumption
\ref{ass:noise-assumption-1} to give 
\begin{align*}
\sum_{j=1}^{n}E|W_{j}|^{2+\delta} & \leq K\sum_{j=1}^{n}\sum_{i=1}^{K}\sigma(\s_{i})^{-(2+\delta)}E|\ve_{j}(\s_{i})|^{2+\delta}\left|\left(\mathbf{A}^{T}\sqrt{\X^{T}\X}\alpha_{i}\right)_{j}\right|^{2+\delta}\\
 & \leq\sum_{j=1}^{n}\sum_{i=1}^{K}CK\left|\left|\left(\mathbf{A}^{T}\sqrt{\X^{T}\X}\alpha_{i}\right)\right|\right|_{\infty}^{2+\delta}
  \leq CK^{2}||\alpha||_{\infty}n\left|\left|\mathbf{A}^{T}\sqrt{\X^{T}\X}\right|\right|_{\infty}^{2+\delta}\ra0,
\end{align*}
as $n\ra\infty$. This concludes the proof of convergence for the finite dimensional distributions.

It remains to show tightness of the sequence $\left\{ \sqrt{\X^{T}\X}\mathbf{A}\sigma(\s)^{-1}\e(\s)\right\} $.
For any block $B\subset S$ we have with $\gamma$ and $\beta$ as
in Assumption \ref{ass:noise-assumption-1} 
\begin{align*}
E\norm{\sqrt{\X^{T}\X}\mathbf{A}(\sigma^{-1}\e)(B)}_{1}^{2+\gamma} & \leq\norm{\sqrt{\X^{T}\X}\mathbf{A}}_{1}^{2+\gamma}\sum_{j=1}^{n}E|(\sigma^{-1}\ve)(B)|^{\gamma}\\
 & \leq Cn\norm{\mathbf{A}^{T}\sqrt{\X^{T}\X}}_{\infty}^{2+\gamma}\lambda(B)^{1+\beta} \leq C'\lambda(B)^{1+\beta},
\end{align*}
which implies tightness (cf. \citet[Thm. 3]{Bickel1971}). 

The statement (b) follows from the fact that 
\begin{align*}
\pi_{n}^{-\nicefrac{1}{2}}\sigma(\s)^{-1}\left(\wh b_{1}(\s)-b_{1}(\s)\right) & =\pi_{n}^{-\nicefrac{1}{2}}e_{1}^{T}(\X^{T}\X)^{-\nicefrac{1}{2}}\sigma(\s)^{-1}\sqrt{\X^{T}\X}\left(\wh{\b}(\s)-\b(\s)\right)\\
 & \ra v^{T}G^{\otimes p}(\s)\stackrel{\mathcal{D}}{=}||v||_{2}G(\s),
\end{align*}
where the weak convergence from the previous part 
is used. 

The last part (c) of the Theorem is a direct application of Theorem \ref{thm:inclusion}, where parts (a) and (b) guarantee that the assumptions are satisfied.
\end{proof}

\begin{proof}[Proof of Proposition \ref{prop:data_CLT}]
In order to be able to apply Theorem \ref{thm:weak-convergence} we compute
\[
\X^{T}\X=\left(\begin{array}{cccc}
\nicefrac{n}{2} & \nicefrac{n}{2} & 0 & 0\\
\nicefrac{n}{2} & n & 0 & 0\\
0 & 0 & \omega^{(a)} & 0\\
0 & 0 & 0 & \omega^{(b)}
\end{array}\right),\quad\omega^{(a)}=\sum_{j=1}^{n^{(a)}}\left(t_{j}^{(a)}\right)^{2},\:\omega^{(b)}=\sum_{j=n^{(a)}+1}^{n^{(a)}+n^{(b)}}\left(t_{j}^{(b)}\right)^{2}.
\]
It follows that 
\[
\left(\X^{T}\X\right)^{-1}=\begin{pmatrix}\nicefrac{4}{n} & -\nicefrac{2}{n} & 0 & 0\\
-\nicefrac{2}{n} & \nicefrac{2}{n} & 0 & 0\\
0 & 0 & \nicefrac{1}{\omega^{(a)}} & 0\\
0 & 0 & 0 & \nicefrac{1}{\omega^{(a)}}
\end{pmatrix},\quad(\X^{T}\X)^{-\nicefrac{1}{2}}=\sqrt{\frac{n}{10}}\begin{pmatrix}\nicefrac{6}{n} & -\nicefrac{2}{n} & 0 & 0\\
-\nicefrac{2}{n} & \nicefrac{4}{n} & 0 & 0\\
0 & 0 & \sqrt{\frac{10}{n\omega^{(a)}}} & 0\\
0 & 0 & 0 & \sqrt{\frac{10}{n\omega^{(b)}}}
\end{pmatrix}.
\]
With this we obtain 
\[
||\X(\X^{T}\X)^{-\nicefrac{1}{2}}||_{\infty}\leq\frac{2}{\sqrt{n}}+\frac{\max_{1\leq j\leq\nicefrac{n}{2}}|t_{j}^{(a)}|}{\sqrt{\omega^{(a)}}}+\frac{\max_{\nicefrac{n}{2}+1\leq j\leq n}|t_{j}^{(b)}|}{\sqrt{\omega^{(b)}}}.
\]
Now we note that since the design points $t_{j}^{(a)}$ and $t_{j}^{(b)}$
are equally spaced by Assumption \ref{ass:appl} we have 
$\max_{1\leq j\leq n}|t_{j}^{(a)}|=\mathcal{O}(n)$ and $\omega^{(a)}=\mathcal{O}(n^{3})$,
and the same is true for the $(b)$-counterparts. This shows that
$||\X(\X^{T}\X)^{-\nicefrac{1}{2}}||_{\infty}=\mathcal{O}(n^{-\nicefrac{1}{2}})$
and therefore Assumptions \ref{ass:appl} imply Assumptions \ref{ass:noise-assumption-1}.
Now, in the notation of Theorem \ref{thm:weak-convergence}, we have $\pi_{n}^{-\nicefrac{1}{2}}=\nicefrac{\sqrt{n}}{2}$
and 
\[
\pi_{n}^{-\nicefrac{1}{2}}e_{1}^{T}(\X^{T}\X)^{-\nicefrac{1}{2}}=\frac{\sqrt{n}}{2}\sqrt{\frac{n}{10}}\begin{pmatrix}\frac{6}{n} & -\frac{2}{n}\end{pmatrix}=\frac{1}{2\sqrt{10}}\begin{pmatrix}6 & -2\end{pmatrix}=:\mathbf{v}^{T},
\]
so that $||\mathbf{v}||_{2}=1$. This finally gives 
$\frac{\sqrt{n}}{2}\sigma(\s)^{-1}(\wh b_{1}(\s)-b_{1}(\s))\ra G(\s)$.
\end{proof}
\end{document}